\documentclass[11pt,letterpaper]{article}

\usepackage[T1]{fontenc}
\usepackage[cp1251]{inputenc}
\usepackage{textcomp}
\usepackage[centertags]{amsmath}
\usepackage{amsfonts}
\usepackage{amssymb}
\usepackage[hypertex]{hyperref}
\usepackage[numbers,sort&compress]{natbib}
\usepackage{amsthm}

\usepackage{paperinitial}

\paperinitialization{15mm}{15mm}{20mm}{10mm}{2pt}{10pt}

\DeclareMathOperator{\re}{Re}
\DeclareMathOperator{\im}{Im}
\DeclareMathOperator{\Tr}{Tr}
\DeclareMathOperator{\tr}{tr}
\DeclareMathOperator{\sgn}{sgn}

\newcommand{\e}{\varepsilon}
\newcommand{\vf}{\varphi}

\newcommand{\s}{\sigma}

\newcommand{\al}{\alpha}
\newcommand{\be}{\beta}
\newcommand{\ga}{\gamma}
\newcommand{\Ga}{\Gamma}
\newcommand{\de}{\delta}

\newcommand{\la}{\lambda}
\newcommand{\La}{\Lambda}
\newcommand{\ups}{\upsilon}

\newcommand{\spx}{\mathbf{x}}

\newcommand{\N}{\mathbb{N}}

\newtheorem{thm}{Theorem}

\begin{document}
\frenchspacing
\allowdisplaybreaks[4]

\title{{\Large \textbf{High-temperature expansion of the one-loop effective action\\ induced by scalar and Dirac particles}}}

\date{}

\author{I.S. Kalinichenko\thanks{E-mail: \texttt{theo@sibmail.com}},\; P.O. Kazinski\thanks{E-mail: \texttt{kpo@phys.tsu.ru}}\\[0.5em]
{\normalsize Physics Faculty, Tomsk State University, Tomsk 634050, Russia}}

\maketitle

\begin{abstract}

The complete nonperturbative expressions for the high-temperature expansion of the one-loop effective action induced by the charged scalar and the charged Dirac particles both at zero and finite temperatures are derived with account for possible nontrivial boundary conditions. The background electromagnetic field is assumed to be stationary and such that the corresponding Klein-Gordon operator or the Dirac Hamiltonian are self-adjoint. The contributions of particles and antiparticles are obtained separately. The explicit expressions for the $C$-symmetric and the non $C$-symmetric vacuum energies of the Dirac fermions are derived. The leading corrections to the high-temperature expansion due to the nontrivial boundary conditions are explicitly found. The corrections to the logarithmic divergence of the effective action that come from the boundary conditions are derived. The high-temperature expansion of the naive one-loop effective action induced by charged fermions turns out to be divergent in the limit of a zero fermion mass.

\end{abstract}

\section{Introduction}

The evaluation of the high-temperature expansions of the one-loop effective actions is a well elaborated procedure at the present moment. As for quantum electrodynamics (QED), the first results in this field date back to the papers \cite{Ditt,ChodEverOw}. Surprisingly, as far as we know, a uniquely defined (in the nonperturbative sense) expression for the complete high-temperature expansion of the one-loop effective action in QED in the Minkowski spacetime is not derived for a general stationary electromagnetic background. There are two issues in obtaining such a high-temperature expansion: i) To treat properly and nonperturbatively the external electric field; ii) To organize the high-temperature expansion in a manageable way that allows one to deduce its complete form. As for the latter problem, we may distinguish the works \cite{HabWeld,Weldon86,ByVaZe,BytsZerb1,BytsZerb2}. In \cite{ByVaZe,BytsZerb1,BytsZerb2}, the high-temperature expansion was obtained but for the case $A_0=0$ and using the Euclidean approach. As for the former problem, we mention the papers and the book \cite{LoewRoj,ElmSkag95,ElmSkag98,Shovkovy,Shovk99,FursVass}, where the leading terms of the high-temperature expansion for charged particles in the electric field were found. In the present paper, we shall derive the complete nonperturbative high-temperature expansion of the one-loop effective action induced by the charged scalars and the Dirac fermions in a stationary electromagnetic field of a general form in the Minkowski spacetime. We shall also find the contribution to this expansion coming from the nontrivial (MIT bag) boundary condition \cite{CJJThW,CJJTh} imposed on the massive Dirac fermions.

When the external stationary electromagnetic field is purely magnetic, the Dirac equation can be squared to the equation of Klein-Gordon (KG) type with a self-adjoint operator. This facilitates the problem of finding the one-loop thermodynamic potential since the powerful methods based on the use of the spectral zeta-function of the Laplace type operators can be employed \cite{HabWeld,Weldon86,ByVaZe,BytsZerb1,BytsZerb2,DaicFr96,KalKaz1,KalKaz2,KalKaz3}. However, in the presence of the electric field, these methods cannot be immediately applied in the case of Dirac fermions because the naive squaring of the Dirac equation leads to the KG type equation with a non-Hermitian operator (see, e.g., \cite{DunHall,DeWGAQFT,BagGit1}). Furthermore, this KG type operator possesses a kernel larger than the initial stationary Dirac equation. Therefore, we have to modify the approaches of \cite{ByVaZe,BytsZerb1,BytsZerb2,DaicFr96,KalKaz1,KalKaz2,KalKaz3,FursVass} to cope this problem. The main idea we shall rest on consists in using the square of the Dirac Hamiltonian instead of the ``square'' of the Dirac equation operator (see \eqref{Dirac_sq2}). This idea is not new and was employed, for example, in \cite{Furs1,Furs2,EliBorKir}. However, as for deriving the one-loop effective action, such a method did not find an application in QED for general background electromagnetic fields. The detailed comparison of this method with other approaches will be given in Sec. \ref{Comparison}, and so we do not dwell on it here.

As is known \cite{ElmSkag95,ElmSkag98,SkalDem,KazShip,gmse}, the high-temperature expansion of the one-loop $\Omega$-potential without the vacuum contribution can be used to obtain the energy of vacuum fluctuations at zero temperature. Therefore, we shall derive a nonperturbative representation of the vacuum energy for the Dirac fermions in a stationary electromagnetic field of a general configuration without resorting to the Wick rotation prescription (as for bosons, see \cite{KalKaz1,KalKaz2,KalKaz3}). This vacuum energy is real-valued for not too wild fields and corresponds to the standard definition of the vacuum of quantum fields on stationary backgrounds (see, e.g., \cite{GrMuRaf}). In fact, there are two definitions of the vacuum energy for Dirac fermions: the $C$-symmetric (see, e.g., \cite{GrMuRaf, FursVass}) and the non $C$-symmetric one (see, e.g., \cite{AkhBerQED,LandLifQED}). We shall find an explicit formula for the difference between the finite parts of these vacuum energies. It turns out that this difference is given by a surface term and disappears when the spectral problem is posed in the whole space, i.e., when the MIT bag boundary condition is not imposed. We shall also find the leading corrections to the high-temperature expansion coming from the MIT bag boundary condition. In the particular case of neutral massless Dirac fermions confined to a sphere, these corrections coincide with the known ones \cite{DeFranc}. We shall obtain the generalization of the coefficient $a_2$ controlling the conformal anomaly to the case of charged massive fermions obeying the MIT bag boundary condition. In the particular case of neutral massive Dirac fermions confined to a sphere, this expression is in agreement with the known one \cite{EliBorKir}. Besides, we shall derive the high-temperature expansion of the separate contributions of particles and antiparticles to the $\Omega$-potential. This expression can be used to obtain the number of particle-antiparticle pairs in the system \cite{KhalilB,KalKaz3} in the high-temperature regime. The leading order terms coincide with that known in the literature \cite{KhalilB}.

In \cite{KalKaz1,KalKaz2}, we developed a general procedure of how to obtain the complete high-temperature expansion when the spectral problem is reduced to the solution of the KG type equation with a self-adjoint operator. The background fields are supposed to be of a general form, in particular, $A_0\neq0$ and the metric is stationary. This method is close to, but not the same as, that proposed in \cite{Furs1,Furs2,FursVass}. We start in Sec. \ref{HTE_Deriv} with a simpler and more rigorous derivation of the general formula of the asymptotic high-temperature expansion of the one-loop $\Omega$-potential than given in \cite{KalKaz2}. The main tool we shall employ is the theorem \ref{anal_int_prop} borrowed from \cite{GSh} and formulated in Appendix \ref{Analyt_Prop}. This theorem describes some analytic properties of the Mellin transform. Using this theorem and the Mellin transforms of the Bose-Einstein and the Fermi-Dirac distribution functions, we shall deduce the desired high-temperature expansion. Notice that the Mellin transform technique for deriving the high-temperature expansion had been already used in \cite{HabWeld,Weldon86}. In Sec. \ref{Omega_Pot_T}, we shall present the derivation of formula (9) of \cite{KalKaz3} for the Fermi-Dirac $\Omega$-potential at zero temperature and nonzero chemical potential. The explicit expression for the nonrenormalized vacuum energy will be also given. Then, in Sec. \ref{HTE_Dir}, the complete high-temperature expansion of the one-loop $\Omega$-potential induced by the charged fermions obeying the nontrivial boundary conditions will be obtained. This expansion is written in terms of the spectral zeta- and eta-functions of the squared Dirac Hamiltonian. The explicit expressions for the $C$-symmetric and the non $C$-symmetric vacuum energies at zero temperature will be given. After that, we shall derive the large mass expansion of the vacuum energy without the surface terms. In Sec. \ref{Lead_Term_HTE}, we shall present the explicit formulas for the leading terms of the high-temperature expansion with account for the MIT bag boundary condition. As for the volume contributions, these leading terms coincide with the known ones \cite{ElmSkag95,ElmSkag98}. In this section, we shall also obtain the renormalized vacuum energy and the high-temperature expansion of the total one-loop effective action induced by Dirac fermions with the vacuum energy included. We shall prove that, at zero chemical potential, this high-temperature expansion is expressed solely in terms of the heat kernel expansion coefficients no matter how small the mass of the field is. Because of the term proportional to the logarithm of a mass, the one-loop effective action at finite temperature diverges in the limit of a zero fermion mass for the charged Dirac fermions in the high-temperature regime.

\section{Derivation of the formula for the high-temperature expansion}\label{HTE_Deriv}

At the beginning, we present several general formulas from \cite{KalKaz2}. Let $\mathcal{K}(\omega)$ be a Fourier transform of the Hermitian operator of KG type. We confine the system at issue into a large box and suppose that $\mathcal{K}(\omega)$ is a self-adjoint operator of Laplace type possessing a spectrum bounded from above at fixed value of $\omega$ in the Hilbert space of square integrable functions meeting appropriate boundary conditions on the surface of the box. Then, the spectrum $\e_k(\omega)$ of such an operator is discrete with an accumulation point at infinity \cite{Agran,Shubin,Gilkey2.11}. Let us introduce the operator
\begin{equation}\label{K_nu}
    \mathcal{K}^{-\nu}_+(\omega):=e^{-i\pi\nu}\Ga(1-\nu)\int_{C}\frac{d\tau\tau^{\nu-1}}{2\pi i}e^{-\tau \mathcal{K}(\omega)},
\end{equation}
where the contour $C$ runs from below upwards slightly to the left and parallel to the imaginary axis, and $\tau^\nu:=|\tau|^\nu e^{i\nu\arg\tau}$, $\arg\tau\in[0,2\pi)$. The operator \eqref{K_nu} is trace-class when $\re\nu<0$. Let $\omega^\al_k$, $\al=\overline{1,n(k)}$, $n(k)<\infty$, be the real solutions to the equation
\begin{equation}
    \e_k(\omega^\al_k)=0,
\end{equation}
and
\begin{equation}\label{stab_1}
    \e'_k(\omega^\al_k)\neq0.
\end{equation}
Then, the function
\begin{equation}\label{zeta_+}
    \zeta_+(\nu,\omega)=\Tr\mathcal{K}^{-\nu}_+(\omega),
\end{equation}
understood as a generalized function of $\omega$, is analytic for $\re\nu<1$ and admits an analytic continuation to the whole complex $\nu$ plane. Let $\vf(\omega)$ be the trial function, then
\begin{equation}
    \int d\omega\vf(\omega) \zeta_+(\nu,\omega)=\int_0^\infty d\e\e^{-\nu}f_\vf(\e),\qquad f_\vf(\e)=-\int d\omega\vf(\omega)\partial_\e\Tr\theta(\mathcal{K}(\omega)-\e),
\end{equation}
where $\Tr\theta(\mathcal{K}(\omega))$ is a generalized function counting the number of positive eigenvalues of $\mathcal{K}(\omega)$. Suppose that $\e_k(\omega)$ is infinitely differentiable at the points $\omega_k^\al$, and the condition \eqref{stab_1} holds true. Then, there is an asymptotic expansion
\begin{equation}
    f_\vf(\e)\simeq\sum_{k=0}^\infty f_k[\vf]\e^k,
\end{equation}
when $\e\rightarrow+0$. Consequently, according to the theorem \ref{anal_int_prop}, the generalized function $\zeta_+(\nu,\omega)$ possesses singularities in a form of simple poles at the points $\nu\in\N$, and the function $\zeta_+(\nu,\omega)/\Ga(1-\nu)$ is an entire function of $\nu$.

Let the additional stability conditions are met \cite{MigdalMM,FursVass,Furshep}:
\begin{equation}
    a)\;\sgn(\omega_k^\al)\e'_k(\omega_k^\al)>0,\qquad b)\;\zeta_+(\nu,\omega)=0,\quad\omega\in(-\epsilon,\epsilon),
\end{equation}
for some $\epsilon>0$. Then, the one-loop $\Omega$-potential takes the form \cite{KalKaz2}
\begin{equation}
    \Omega=-\int_0^\infty d\omega\Big[\frac{\zeta_+(0,\omega)}{e^{\be(\omega-\mu)}\pm1} +\frac{\zeta_+(0,-\omega)}{e^{\be(\omega+\mu)}\pm1}\Big].
\end{equation}
The first term in this expression is a contribution from particles, and the second one corresponds to antiparticles.

The theorem \ref{anal_int_prop} allows us to give a simple proof of the general formula for the high-temperature expansion of the $\Omega$-potential \cite{KalKaz1,KalKaz2}. Recall the main assumptions made about the zeta-function:
\begin{enumerate}
  \item $\zeta_+(\nu,\omega)=0$ for $\omega\in[0,\omega_c]$;
  \item $\zeta_+(\nu,\omega)$ is absolutely locally integrable on the ray $[\omega_c,+\infty)$ (when $\re\nu<1$);
  \item There is an asymptotic expansion when $\omega\rightarrow+\infty$:
  \begin{equation}\label{zeta_asympt}
    \zeta_+(\nu,\omega)=\sum_{k=0}^N\zeta^{+}_k(\nu)\omega^{d-2\nu-k}+O(\omega^{d-2\nu-N-1}).
  \end{equation}
\end{enumerate}
As was shown in \cite{KalKaz1,KalKaz2} by a direct calculation, the third property holds for the operator $\mathcal{K}(\omega)$ being a Fourier transform of a KG type operator. The coefficients $\zeta_k^+(\nu)$ are expressed through the expansion coefficients of the heat kernel appearing in \eqref{K_nu} (see, e.g., \cite{KalKaz1,KalKaz2,Furs1,Furs2,FursVass}). The second property follows from the definition \eqref{zeta_+}.

Let us start with the case of the Bose-Einstein distribution and introduce the notation
\begin{equation}
    I_{\nu}(\mu)=\int_{0}^{\infty} \frac{d\omega\zeta_{+}(\nu,\omega)}{e^{\beta(\omega-\mu)}-1},\qquad \re\nu<1,
\end{equation}
where $\mu\in(0,\omega_c)$. Using the Mellin transform, we replace a nontrivial functional dependence on $\omega$ in the integrand by a simple power law
\begin{equation}
    \frac{1}{e^{\beta(\omega-\mu)}-1}=\int_{C_1}\frac{d s}{2\pi i}\Gamma(-s)\zeta(-s)[\beta(\omega-\mu)]^s,\quad \re(\omega-\mu)>0,
\end{equation}
where the contour $C_1$ runs parallel to the imaginary axis from below upwards slightly to the left of the point $s=-1$. Then,
\begin{equation}\label{Ioversigma}
    I_{\nu}(\mu)=\int_{C_1}\frac{d s}{2\pi i}\Gamma(-s)\zeta(-s)\beta^s\sigma^{s}_{\nu}(\mu),
\end{equation}
where
\begin{equation}
    \sigma^{s}_{\nu}(\mu):=\int_{0}^{\infty} d\omega(\omega-\mu)^s\zeta_{+}(\nu,\omega)=\int_{\omega_c}^{\infty} d\omega(\omega-\mu)^s\zeta_{+}(\nu,\omega).
\end{equation}
Formula \eqref{Ioversigma} is valid only when the contour $C_1$ is shifted to the domain of convergence of the integral representation for $\sigma^{s}_{\nu}(\mu)$, i.e., when $\re s<2\re \nu-d-1$.

In order to derive the high-temperature expansion of \eqref{Ioversigma}, we have to shift the contour $C_1$ to the domain of positive $\re s$ and take into account the pole structure of the integrand. The location of the poles of the function $\Gamma(-s)\zeta(-s)$ is evident, and the structure of singularities of the function $\sigma^{s}_{\nu}(\mu)$ can be determined by means of the theorem \ref{anal_int_prop}. The change of the variable, $x=(\omega-\mu)^{-1}$, transforms the integral to the desired form
\begin{equation}
    \sigma^{s}_{\nu}(\mu)=\int^{(\omega_c-\mu)^{-1}}_{0} dx x^{-s-2}\zeta_{+}(\nu,\mu+x^{-1}).
\end{equation}
It is not difficult to find the asymptotic expansion of the integrand,
\begin{equation}\label{ammu}
    x^{2\nu-d-s-2}\Big[\sum_{m=0}^{N}a_m(\mu)x^m+O(x^{N+1})\Big],\qquad a_m(\mu)=\sum_{n=0}^{m}\zeta^{+}_{m-n}(\nu)\frac{\Gamma(D-2\nu-m+n)}{\Gamma(D-2\nu-m)}\frac{\mu^n}{n!},
\end{equation}
where $D:=d+1$. All the theorem \ref{anal_int_prop} conditions are satisfied. Therefore, the function $\sigma^{s}_{\nu}(\mu)$ can be continued analytically to the domain $\re s<2\re\nu-d+N$, where
\begin{equation}\label{sigma_gen2}
    \s^s_\nu(\mu)=\int_0^{(\omega_c-\mu)^{-1}}dxx^{-s-2}\Big[\zeta(\nu,\mu+x^{-1}) -\sum_{m=0}^N a_m(\mu)x^{2\nu-d+m}\Big]-\sum_{m=0}^N\frac{a_m(\mu)(\omega_c-\mu)^{s+D-2\nu-m}}{s+D-2\nu-m}.
\end{equation}
It possesses simple poles at the points $s=2\nu-d-2+k$, $k=\overline{1,N+1}$, with the residues $-a_{k-1}(\mu)$. The implication of the theorem \ref{anal_int_prop},
\begin{equation}\label{asympt_sigma}
    \lim_{|\im s|\rightarrow\infty}\s^s_\nu(\mu)=0,\qquad \re s<2\re\nu-d+N,
\end{equation}
guarantees the convergence of \eqref{Ioversigma} and allows us to shift the contour $C_1$ to the right.

Hence, we obtain
\begin{equation}\label{HTE_gen}
\begin{split}
    I_{\nu}(\mu)=&\sum_{m=0}^N\be^{m+2\nu-D}\zeta(D-2\nu-m)\sum_{n=0}^m\Ga(D-2\nu-m+n)\zeta^{+}_{m-n}(\nu)\frac{\mu^n}{n!}+\\
    &+\sum_{l=-1}^{l_0}\frac{(-1)^l\zeta(-l)}{\Ga(l+1)}\s^l_\nu(\mu)\be^l+\int_{C_1}\frac{ds}{2\pi i}\Ga(-s)\zeta(-s)\s^s_\nu(\mu)\be^s,
\end{split}
\end{equation}
where $l_0=\lfloor 2\re\nu-d+N \rfloor$. Notice that the term with $l=-1$ is understood as a limit, and the contour $C_1$ now runs slightly to the left of the line $\re s=N+2\re \nu-d$.

If $N\rightarrow\infty$, then setting aside an exponentially suppressed, as $\be\rightarrow+0$, term, we deduce from \eqref{HTE_gen} the asymptotic expansion
\begin{equation}\label{expan_bos}
    -\Omega_{b}(\mu)\simeq\sum_{k,n=0}^\infty\Gamma(D-2\nu-k)\zeta(D-2\nu-k-n)\frac{\zeta^+_{k}(\nu)(\beta\mu)^n}{n!\beta^{D-2\nu-k}}
    +\sum^{\infty}_{l=-1}\frac{(-1)^l \zeta(-l)}{\Gamma(l+1)}\sigma^l_{\nu}(\mu)\beta^l,\quad \nu\rightarrow0.
\end{equation}
As for fermions, the considerations are quite analogous, but with the difference that the Riemann zeta-function $\zeta(z)$ should be replaced by the Dirichlet eta-function $\eta(z):=(1-2^{1-z})\zeta(z)$ in all the appearances. Therefore,
\begin{equation}\label{expan_ferm}
    -\Omega_{f}(\mu)\simeq\sum_{k,n=0}^\infty\Gamma(D-2\nu-k)\eta(D-2\nu-k-n)\frac{\zeta^+_{k}(\nu)(\beta\mu)^n}{n!\beta^{D-2\nu-k}}
    +\sum^{\infty}_{l=0}\frac{(-1)^l \eta(-l)}{\Gamma(l+1)}\sigma^l_{\nu}(\mu)\beta^l,\quad \nu\rightarrow0,
\end{equation}
The term with $l=-1$ is absent in the expansion since the eta-function does not possess a singularity at $z=1$ unlike the zeta-function. The contributions at the even positive $l$ in the second term in \eqref{HTE_gen}, \eqref{expan_bos}, and \eqref{expan_ferm} vanish. With the aid of formulas \eqref{ammu}, \eqref{sigma_gen2} for the singularities of $\s^s_\nu(\mu)$, one can also check directly that expressions \eqref{expan_bos}, \eqref{expan_ferm} are finite when $\nu\rightarrow0$.

The contribution from antiparticles to the high-temperature expansion of the $\Omega$-potential is derived similarly. Let
\begin{equation}\label{zeta_asympt_a}
    \zeta_+(\nu,-\omega)=\sum_{k=0}^N\zeta^-_k(\nu)\omega^{d-2\nu-k}+O(\omega^{d-2\nu-N-1}),\qquad \omega\rightarrow+\infty.
\end{equation}
If the operator $\mathcal{K}(\omega)$ is invariant under the simultaneous substitution $\omega\rightarrow-\omega$, $A_0\rightarrow-A_0$, then
\begin{equation}
    \zeta_+(\nu,-\omega;A_0)=\zeta_+(\nu,\omega;-A_0),
\end{equation}
and
\begin{equation}
    \zeta_k^-(\nu;A_0)=\zeta_k^+(\nu;-A_0).
\end{equation}
Introduce the following functions
\begin{equation}
    \tau^s_\nu(\mu):=\int_{0}^\infty d\omega(\omega+\mu)^s\zeta_+(\nu,-\omega),\qquad\re s<2\re\nu-d-1,
\end{equation}
understood in the sense of analytical continuation when $\re s\geq 2\re\nu-d-1$. Then,
\begin{equation}\label{HTE_gen_a}
\begin{split}
    \int_0^\infty\frac{d\omega\zeta_+(\nu,-\omega)}{e^{\be(\omega+\mu)}-1}=&\sum_{m=0}^N\be^{m+2\nu-d-1}\zeta(d+1-2\nu-m)\sum_{n=0}^m\Ga(d+1-2\nu-m+n)\zeta^-_{m-n}(\nu)\frac{(-\mu)^n}{n!}+\\
    &+\sum_{l=-1}^{l_0}\frac{(-1)^l\zeta(-l)}{\Ga(l+1)}\tau^l_\nu(\mu)\be^l+\int_{C_1}\frac{ds}{2\pi i}\Ga(-s)\zeta(-s)\tau^s_\nu(\mu)\be^s.
\end{split}
\end{equation}
The proof of this formula is absolutely analogous to the one given above for particles.

\section{Omega-potential of fermions at $T=0$}\label{Omega_Pot_T}

In this section, we derive formula (9) of \cite{KalKaz3} for the nonrenormalized one-loop $\Omega$-potential of fermions at zero temperature and nonzero chemical potential with the contribution of the energy of zero-point fluctuations. The formula may be of use when the solution of the Dirac equation reduces to the solution of KG equation with a self-adjoint operator.

There are two definitions of the energy of vacuum fluctuations for the Dirac fermions in the literature: the $C$-symmetric and the non $C$-symmetric one (see \cite{GrMuRaf} for details). Let us start with the nonsymmetric one. According to the nonsymmetric definition (see, e.g., \cite{AkhBerQED,LandLifQED}), the energy of vacuum fluctuations of fermions is the energy of the ``Dirac sea''. Then, for fermions at zero temperature but nonzero chemical potential,
\begin{equation}\label{Evac_as0}
    E_{vac}=-\int_0^{\Lambda}d\omega\omega\rho(-\omega),\qquad E_{part}=\int_0^{-\mu}d\omega\omega\rho(-\omega),
\end{equation}
where $\La$ is a cut-off parameter, $E_{vac}$ is a nonrenormalized vacuum energy, $E_{part}$ is an average energy of particles, and
\begin{equation}
    \rho(-\omega)=\sgn(\omega)\partial_\omega\Tr\theta(\mathcal{K}(-\omega))
\end{equation}
is the spectral density of energies of antiparticles. Therefore,
\begin{equation}
    E_{tot}=E_{vac}+E_{part}=-\int_{-\mu}^\La d\omega\omega\rho(-\omega)\underset{\be_0\rightarrow0}{\longleftarrow}-2\int_{-\mu}^\infty\frac{d\omega\omega\rho(-\omega)}{e^{\be_0\omega}+1},
\end{equation}
where $\be_0$ is some regularization parameter. The last equality is accurate within the renormalization ambiguity. Thus,
\begin{equation}
\begin{split}
    E_{tot}&=2\frac{\partial}{\partial\be_0}\Big[\be_0\int_{-\mu}^\infty d\omega\sgn(\omega)\frac{\Tr\theta(\mathcal{K}(-\omega))}{e^{\be_0\omega}+1}+\sgn(\mu)\Tr\theta(\mathcal{K}(\mu))\ln(1+e^{-\be_0\mu}) \Big]_{\be_0\rightarrow0}=\\
    &=2\frac{\partial}{\partial\be_0}\Big[\be_0\int_{-\mu}^\infty d\omega\sgn(\omega)\frac{\Tr\theta(\mathcal{K}(-\omega))}{e^{\be_0\omega}+1}\Big]_{\be_0\rightarrow0}+|\mu|\Tr\theta(\mathcal{K}(\mu)),
\end{split}
\end{equation}
where it is assumed that $\Tr\theta(\mathcal{K}(0))=0$. In this case, $\Tr\theta(\mathcal{K}(\mu))$ for $\mu>0$ gives the number of states with a frequency $0<\omega<\mu$ and, for $\mu<0$, it gives the number of states with a frequency $\mu<\omega<0$. Assigning $+1$ to the charge of particles and $-1$ to the charge of antiparticles, we obtain
\begin{equation}\label{omega_pot_zero_temp}
    \Omega_{tot}=E_{tot}-\mu Q=2\frac{\partial}{\partial\be_0}\Big[\be_0\int_{-\mu}^\infty d\omega\sgn(\omega)\frac{\Tr\theta(\mathcal{K}(-\omega))}{e^{\be_0\omega}+1}\Big]_{\be_0\rightarrow0},
\end{equation}
i.e., we arrive at formula (9) of \cite{KalKaz3}.

Applying the formulas of Sec. \ref{HTE_Deriv} to the expression \eqref{omega_pot_zero_temp}, we deduce
\begin{equation}
\begin{split}
    \int_{-\mu}^\infty d\omega\sgn(\omega)\frac{\Tr\theta(\mathcal{K}(-\omega))}{e^{\be_0\omega}+1}=&\sum_{m=0}^N\be^{m+2\nu-D}_0\Ga(D-2\nu-m)\eta(D-2\nu-m)\zeta_{m}^-(\nu)+\\
    &+\sum_{l=0}^{l_0}\frac{(-1)^l\eta(-l)}{\Ga(l+1)}\tilde{\s}^l_\nu(\mu)\be^l_0+\int_{C_1}\frac{ds}{2\pi i}\Ga(-s)\eta(-s)\tilde{\s}^s_\nu(\mu)\be^s_0,\quad\nu\rightarrow0,
\end{split}
\end{equation}
where
\begin{equation}\label{sigma_tilde}
    \tilde{\s}^s_\nu(\mu):=\int_{-\mu}^\infty d\omega\sgn(\omega)\omega^s\zeta_+(\nu,-\omega).
\end{equation}
To justify the applicability of the formulas of Sec. \ref{HTE_Deriv}, we have to assume that $\zeta_+(\nu,\omega)=0$ in some neighborhood of $\omega=0$ and partition the integral in \eqref{sigma_tilde} into two:
\begin{equation}
    \tilde{\s}^s_\nu(\mu)=\int_{-\mu}^0 d\omega\sgn(\omega)\omega^s\zeta_+(\nu,-\omega)+\int_0^\infty d\omega\sgn(\omega)\omega^s\zeta_+(\nu,-\omega).
\end{equation}
The consideration of Sec. \ref{HTE_Deriv} is to be applied to the second integral, while the first integral is an entire function of $s$ tending to zero at $|\im s|\rightarrow\infty$.

If we are interested only in the logarithmic divergence and the finite part of \eqref{omega_pot_zero_temp} as $\be_0\rightarrow0$, then
\begin{equation}
    \Omega_{tot}=(4\nu+2)\Ga(-2\nu)\eta(-2\nu)\zeta^-_{D}(\nu)\be_0^{2\nu}+\tilde{\s}^0_\nu(\mu)+\cdots,\quad\nu\rightarrow0,
\end{equation}
where $\be_0$ tends to zero after the limit $\nu\rightarrow0$. Evaluating this limit, we come to
\begin{equation}\label{omega_tot_mu}
    -\Omega_{tot}=\frac12\Big[\zeta_D^-(0)\ln\frac{4\be_0^2e^{2\ga+2}}{\pi^2}+ \partial_\nu\zeta_D^-(0)\Big]-\text{f.p.}\,\int_0^\infty d\omega\zeta_+(\nu,-\omega)\big|_{\nu\rightarrow0}+\int_0^\mu d\omega\sgn(\omega)\zeta_+(0,\omega),
\end{equation}
where $\text{f.p.}$ denotes the finite part of the expression as $\nu\rightarrow0$. The expression \eqref{omega_tot_mu} needs to be renormalized as, e.g., in \eqref{Evacren}.

When the $C$-symmetric definition for the vacuum energy is used (see, for example, \cite{GrMuRaf,FursVass}),
\begin{equation}\label{Evac1s}
    E^c_{vac}=-\frac12\sum_k E^{(-)}_k -\frac12\sum_k E^{(+)}_k.
\end{equation}
Such an expression for the vacuum energy originates from the transformation of the initially Weyl-ordered Hamiltonian of fermionic fields to the normal form. Then, the vacuum energy of the Dirac fermions reads as
\begin{equation}
    E^c_{vac}=\frac{\partial}{\partial\be_0}\Big\{\sum_{k=0}^\infty\Ga(D-2\nu-k)\eta(D-2\nu-k)\frac{\zeta_k^+(\nu)+\zeta_k^-(\nu)}{\be_0^{d-2\nu-k}} +\sum_{l=0}^\infty\frac{(-1)^l\eta(-l)}{\Ga(l+1)}\be_0^{l+1}\big[\s^l_\nu(0)+\tau^l_\nu(0)\big]\Big\}_{\be_0\rightarrow0},
\end{equation}
where $\nu\rightarrow0$. Taking the limit $\nu\rightarrow0$ and keeping only the finite and divergent terms when $\be_0\rightarrow0$, we arrive at the nonrenormalized $C$-symmetric vacuum energy
\begin{multline}
    E^c_{vac}=\sum_{k=0}^{d-1}\Ga(D-k)\eta(D-k)\frac{\zeta_k^+(0)+\zeta_k^-(0)}{\be_0^{D-k}}(k-d)-\\
    -\frac12\Big\{\frac{\zeta^+_D(0)+\zeta^-_D(0)}{2}\ln\frac{4\be_0^2e^{2\ga+2}}{\pi^2} +\frac{\partial_\nu\zeta_D^+(0)+\partial_\nu\zeta_D^-(0)}{2} -\text{f.p.}\big[\s^0_\nu(0)+\tau^0_\nu(0)\big]_{\nu\rightarrow0}\Big\}.
\end{multline}
Notice that
\begin{equation}
    \s^0_\nu(0)+\tau^0_\nu(0)=\int_{-\infty}^\infty d\omega \zeta_+(\nu,\omega)= e^{-i\pi\nu}\Ga(1-\nu)\int_{-\infty}^\infty d\omega\int_{C}\frac{d\tau\tau^{\nu-1}}{2\pi i}\Tr e^{-\tau \mathcal{K}(\omega)}.
\end{equation}
However,
\begin{equation}\label{pseudoheat_kern}
    T\int_{-\infty}^\infty \frac{d\omega}{2\pi} \Tr e^{-\tau \mathcal{K}(\omega)}\neq \Tr_4 e^{-\tau \mathcal{K}(i\partial_t)},
\end{equation}
where $T$ is an observation period, for the operator under the trace sign on the right-hand side is not trace-class. The left-hand side of \eqref{pseudoheat_kern} can be regarded as the definition for the right-hand side of \eqref{pseudoheat_kern} in the case of stationary background fields. The one-loop correction to the total nonrenormalized $\Omega$-potential is written as
\begin{equation}
    \Omega^c_{tot}=E^c_{vac}-\int_0^\mu d\omega\sgn(\omega)\zeta_+(0,\omega).
\end{equation}
The one-loop correction to the effective Lagrangian, $L_{eff}^{(1)}=-\Omega^c_{tot}$, is real and corresponds to the standard choice of the vacuum state for stationary background fields.

\section{High-temperature expansion for the Dirac fermions}\label{HTE_Dir}

In the previous sections, we have derived the formulas for the high-temperature expansion of the one-loop $\Omega$-potential induced by fermions in the case when the spectral problem for the Dirac equation can be reduced to solving the KG type equation with a self-adjoint operator. This situation is realized in the absence of the external electric field (see, e.g., \cite{KhalilB,Drozd,Dunne2,AndNayTran,KalKaz3} and references therein) and for certain special configurations of the electromagnetic fields (see \cite{BagGit1}).

In this section, we shall derive the formula for the high-temperature expansion of the $\Omega$-potential induced by fermions in the external stationary electromagnetic fields of a general configuration subject to the two restrictions:
\begin{enumerate}
  \item The external electromagnetic field is such that the Dirac Hamiltonian is self-adjoint;
  \item The Dirac Hamiltonian does not possess zero modes.
\end{enumerate}
The first requirement says that the electromagnetic field does not have too strong singularities, and the electromagnetic potentials do not grow too fast at spatial infinity (if the problem is posed in the entire space). The second requirement can be violated in superstrong electromagnetic fields (see, e.g., \cite{GrMuRaf} for details). This condition is rather technical and can be relaxed because the zero modes do not contribute to the one-loop thermodynamic potential of fermions (see \eqref{HTE_Mellin}).

The eigenvalue eigenvector problem takes the form
\begin{equation}\label{Sturm_Lioville_probl}
    H_D\psi_k=\omega_k\psi_k,
\end{equation}
where $\psi_k$ is a Dirac bispinor and
\begin{equation}
    H_D:=A_0+m\ga^0-\al^i P_i,\qquad\al^i:=\ga^0\ga^i,
\end{equation}
and $P_i=p_i-A_i$, $p_i=i\partial_i$. We assume that the chemical potential $\mu$ conjugate to the electric charge is included into the definition of the potential $A_0$ (see, e.g., \cite{Kapustp}). With this definition, the sign of the chemical potential is opposite to that used in Secs. \ref{HTE_Deriv}, \ref{Omega_Pot_T} and in \cite{KalKaz3}. The bispinor $\psi_k$ obeys the boundary condition \cite{CJJThW,CJJTh,Vassil}
\begin{equation}\label{bound_conds}
    \Pi\psi_k\Big|_{\Ga}=0,
\end{equation}
where $\Ga$ is a smooth boundary of the domain where the problem \eqref{Sturm_Lioville_probl} is posed. The matrix $\Pi$ is a projector, and we choose it in the form corresponding to the so-called MIT bag boundary condition \cite{CJJThW,CJJTh}
\begin{equation}
    \Pi:=\frac{1-i n_\mu\ga^\mu}{2},
\end{equation}
where $n^\mu$ is the inward unit normal to the surface $\Ga$.

Let us introduce the spectral functions (see, e.g., \cite{Shubin,Agran,Gilkey2.11})
\begin{equation}\label{ze_eta}
    \zeta_s(\nu)\equiv\zeta(\nu,H_D):=\sum_k|\omega_k|^{-2\nu},\qquad \eta_s(\nu)\equiv\eta(\nu,H_D):=\sum_k\omega_k|\omega_k|^{-2\nu},
\end{equation}
where we suppose that $\ker H_D=0$. Notice that we use a nonstandard notation for the spectral functions \eqref{ze_eta}. The series defining $\zeta_s(\nu)$, $\zeta'_s(\nu)$ converge absolutely when $\re\nu>(d-1)/2$, while for $\eta_s(\nu)$, $\eta'_s(\nu)$ they converge absolutely when $\re\nu>d/2$. Consequently, these functions are analytic in these domains. For other $\nu$, these functions are understood in the sense of analytic continuation. They possess singularities in the $\nu$ plane in the form of simple poles lying on the real axis (see, e.g., \cite{Shubin,Agran,Gilkey2.11} and below).

The spectral functions introduced allow one to evaluate the integrals of the form
\begin{equation}\label{int_spec}
    \int_0^\infty d\omega\rho(\omega)f(\omega),
\end{equation}
where $\rho(\omega)$ is the spectral density of the Hamiltonian $H_D$, provided that $f(\omega)$ can be expressed in the form of the inverse Mellin transform
\begin{equation}
    f(\omega)=\int_C\frac{ds}{2\pi i}\omega^{s-1}\tilde{f}(s),
\end{equation}
with the contour $C$ running downwards parallel to the imaginary axis at $\re s<2-d$. The integral \eqref{int_spec} can be written in the form
\begin{equation}\label{Mellin_trans}
    \int_C\frac{ds}{2\pi i}\tilde{f}(s)\int_0^\infty d\omega\omega^{s-1}\rho(\omega)=\int_C\frac{ds}{4\pi i}\tilde{f}(s)\Big[\zeta_s\Big(\frac{1-s}{2}\Big)+\eta_s\Big(1-\frac{s}2\Big)\Big],
\end{equation}
where we have substituted the Mellin transform of the spectral density
\begin{equation}
    \int_0^\infty d\omega\omega^{s-1}\rho(\omega)=\frac12\Big[\zeta_s\Big(\frac{1-s}{2}\Big)+\eta_s\Big(1-\frac{s}2\Big)\Big],\qquad \re s<2-d.
\end{equation}
It is assumed that the integral over $s$ on the right-hand side of formula \eqref{Mellin_trans} converges. The representation \eqref{Mellin_trans} makes it possible to find the high-temperature expansion of the $\Omega$-potential in terms of the spectral functions $\zeta_s(\nu)$ and $\eta_s(\nu)$.

Indeed, in virtue of the formula
\begin{equation}
    \ln(1+e^{-\be\omega})=-\int_C\frac{ds}{2\pi i}\frac{\Ga(2-s)\eta(2-s)}{s-1}(\be\omega)^{s-1},\qquad\re s<1,
\end{equation}
it follows from \eqref{Mellin_trans} that the contribution of particles to the one-loop $\Omega$-potential can be written as
\begin{equation}\label{HTE_Mellin}
    -\be\Omega=\int_0^\infty d\omega\rho(\omega) \ln(1+e^{-\be\omega})=-\int_C\frac{ds}{4\pi i}\frac{\Ga(2-s)\eta(2-s)}{s-1}\Big[\zeta_s\Big(\frac{1-s}{2}\Big)+\eta_s\Big(1-\frac{s}2\Big)\Big]\be^{s-1},
\end{equation}
where $\re s<2-d$.

In order to find the expansion in the rising powers of $\be$, we need to move the contour $C$ to the right up to the required power of $\be$ and take into account the singularities of the integrand in the complex $s$ plane. In general, the integral along $C$ moved to $+\infty$ does not tend to zero. Therefore, when the contour $C$ is moved to $+\infty$, the resulting series in the rising powers of $\beta$ is only asymptotic. The integral along $C$ gives the remainder of this expansion and is exponentially suppressed when $\be\rightarrow+0$.

Let us investigate the singularities of $\zeta_s(\nu)$ and $\eta_s(\nu)$ in the $\nu$ plane and their behavior when $|\im\nu|\rightarrow\infty$ (see, e.g., \cite{Agran,BranGilk}). These spectral functions can be expressed in terms of the heat kernel trace
\begin{equation}\label{ze_eta_HK}
    \zeta_s(\nu)=\int_0^\infty\frac{d\tau\tau^{\nu-1}}{\Ga(\nu)}\Tr e^{-\tau H^2_D},\qquad \eta_s(\nu)=\int_0^\infty\frac{d\tau\tau^{\nu-1}}{\Ga(\nu)}\Tr(H_De^{-\tau H^2_D}).
\end{equation}
It is useful to introduce \cite{BranGilk}
\begin{equation}\label{zeta_nue}
    \zeta_s(\nu,\e)=\int_0^\infty\frac{d\tau\tau^{\nu-1}}{\Ga(\nu)}\Tr e^{-\tau(H_D-\e)^2}.
\end{equation}
Then
\begin{equation}\label{eta-zeta}
    \eta_s(\nu+1)=\frac{\Ga(\nu)}{2\Ga(\nu+1)}\frac{\partial}{\partial\e}\zeta_s(\nu,\e)\Big|_{\e=0}=\frac{1}{2\nu}\frac{\partial}{\partial\e}\zeta_s(\nu,\e)\Big|_{\e=0}.
\end{equation}
Hence, it is sufficient to study the analytic properties of $\zeta_s(\nu,\e)$ at $\e\rightarrow0$.

Since $(H_D-\e)^2$ is a Laplace type operator (see \eqref{Dirac_sq}), the following asymptotic expansion holds for $\tau\rightarrow+0$ (see \cite{Vassil} for a review)
\begin{equation}\label{HKE}
    \Tr e^{-\tau(H_D-\e)^2}\simeq\sum_{k=0}^\infty\frac{\tau^{(k-d)/2}}{(4\pi)^{d/2}}a_{k/2}(\e),
\end{equation}
where $a_{k/2}(\e)$ are the heat kernel expansion coefficients. The integrand of \eqref{zeta_nue} is absolutely integrable. On the upper limit, the integral over $\tau$ in \eqref{zeta_nue} is convergent for any $\nu$ because $\ker H^2_D=0$. Therefore, introducing some cutoff $\La$ on the upper integration limit in \eqref{zeta_nue}, we can apply the theorem \ref{anal_int_prop}. As a result, we obtain the structure of singularities of $\zeta_s(\nu)$:
\begin{equation}
    \zeta_s\Big(\frac{1-s}{2}\Big)=\sum_{k=0}^\infty\frac{-2a_{k/2}}{(4\pi)^{d/2}\Ga\big((d-k)/2\big)}\frac{1}{s+d-k-1}+\frac{\text{regular}}{\Ga\big((1-s)/2\big)},
\end{equation}
where $a_{k/2}:=a_{k/2}(0)$. From \eqref{eta-zeta} we have
\begin{equation}
    \eta_s\Big(1-\frac{s}{2}\Big)=\sum_{k=0}^\infty\frac{-a'_{k/2}}{(4\pi)^{d/2}\Ga\big(1+(d-k)/2\big)}\frac{1}{s+d-k}+\frac{\text{regular}}{\Ga(1-s/2)},
\end{equation}
where $a'_{k/2}:=\partial_\e a_{k/2}(\e)$ at $\e=0$. It follows from these formulas that
\begin{equation}\label{zeta_part_v}
    \zeta_s(-n)=(-1)^nn!\frac{a_{n+d/2}}{(4\pi)^{d/2}},\qquad\eta_s(-n)=(-1)^nn!\frac{a'_{n+1+d/2}}{2(4\pi)^{d/2}},\qquad n=\overline{0,\infty}.
\end{equation}
The theorem \ref{anal_int_prop} also implies
\begin{equation}
    \Ga(\nu)\zeta_s(\nu,\e)\underset{|\im \nu|\rightarrow\infty}\longrightarrow0.
\end{equation}
Consequently, when $|\im\nu|\rightarrow\infty$, the functions $|\zeta_s(\nu)|$ and $|\eta_s(\nu)|$ tend to infinity not faster than $e^{|\im\nu|(\pi/2+0)}$. Notice that, in the case we consider, a stronger estimate even holds (see the remark after theorem 5.5.2 in \cite{Agran}). This means that the integral over $s$ on the right-hand side of \eqref{HTE_Mellin} converges, and the contour $C$ can be moved to the right.

Now we see that the integrand of \eqref{HTE_Mellin} has the pole singularities at the points $s=\overline{-d,\infty}$. The poles at $\re s\leq0$ and $s=2l+2$, $l=\overline{1,\infty}$, are simple. The other ones are of the second order. Let us introduce the notation
\begin{equation}
    \bar{\zeta}_s(\nu)=\frac{\partial}{\partial\s}[\s\zeta_s(\nu+\s)]_{\s=0},\qquad \bar{\eta}_s(\nu)=\frac{\partial}{\partial\s}[\s\eta_s(\nu+\s)]_{\s=0}.
\end{equation}
The functions $\bar{\zeta}_s(\nu)$, $\bar{\eta}_s(\nu)$ are the finite part of the Laurent series (i.e., the coefficient $c_0$) of the functions $\zeta_s(\nu)$, $\eta_s(\nu)$ in the vicinity of the point $\nu$. Then the high-temperature expansion of the contribution of particles to the one-loop $\Omega$-potential is
\begin{equation}\label{HTE_ferm}
\begin{split}
    -\be\Omega\simeq&\sum_{k=0}^{d-1}\frac{\Ga(D-k)\eta(D-k)}{2(4\pi)^{d/2}(d-k)}\be^{k-d}\Big[\frac{a'_{(k+1)/2}}{\Ga\big((D-k)/2\big)}+\frac{2 a_{k/2}}{\Ga\big((d-k)/2\big)} \Big]+\\
    &+\ln(\sqrt{2}) \Big[\zeta_s(0)+\bar{\eta}_s(1/2) -\ln(2\be^2) \frac{a'_{D/2}}{(4\pi)^{D/2}}\Big]-\frac{\be}{4}\Big[\bar{\zeta}_s(-1/2)+\eta_s(0) +\frac{a_{D/2}}{(4\pi)^{D/2}}\ln\frac{4\be^2e^{2\ga-2}}{\pi^2} \Big]+\\
    &+\sum_{l=1}^\infty \frac{\eta(1-2l)}{4l\Ga(2l)}\be^{2l}\Big[\zeta_s(-l)+\bar{\eta}_s(1/2-l) -\frac{a'_{l+D/2}}{\Ga(1/2-l)}\frac{\ln(\be\al_l/\pi)}{(4\pi)^{d/2}}\Big]-\\
    &-\sum_{l=1}^\infty\frac{\Ga(-2l)\eta(-2l)}{(4\pi)^{d/2}(2l+1)}\be^{2l+1}\frac{a_{l+D/2}}{\Ga(-1/2-l)}\Big],
\end{split}
\end{equation}
where $D:=d+1$,
\begin{equation}
    \ln\al_l:=\frac{\zeta'(2l)}{\zeta(2l)}+\frac{\ln2}{4^l-1}-\frac{1}{2l},
\end{equation}
and we have taken into account that $a'_{0}=0$. As long as
\begin{equation}
    \int_0^\infty d\omega\omega^{s-1}\rho(-\omega)=\frac12\Big[\zeta_s\Big(\frac{1-s}{2}\Big)-\eta_s\Big(1-\frac{s}2\Big)\Big],\qquad \re s<2-d,
\end{equation}
the contribution of antiparticles has the same form as \eqref{HTE_ferm} but with the replacement $\eta_s(\nu)\rightarrow-\eta_s(\nu)$, $\bar{\eta}_s(\nu)\rightarrow-\bar{\eta}_s(\nu)$, and $a'_k\rightarrow-a'_k$. Therefore, the total $\Omega$-potential does not contain the functions $\eta_s(\nu)$, $\bar{\eta}_s(\nu)$, and the heat kernel expansion coefficients $a'_k$. Bearing in mind \eqref{zeta_part_v}, we see that the asymptotic expansion in $\be$ of the total $\Omega$-potential for fermions without the vacuum energy contribution is expressed through the heat kernel expansion coefficients save the term proportional to $\bar{\zeta}_s(-1/2)$.

Let us step back for a while and consider \eqref{HTE_Mellin} once again. If we substitute the representation \eqref{ze_eta_HK} into \eqref{HTE_Mellin} and change the order of integration over $s$ and $\tau$, the following integral arises
\begin{equation}
    I(\be,\tau):=-\tau^{-1}\int_C\frac{ds}{4\pi i}\frac{\Ga(2-s)\eta(2-s)}{(s-1)\Ga\big((1-s)/2\big)}\Big(\frac{\be^2}{\tau}\Big)^{(s-1)/2}.
\end{equation}
This integral is reduced to the Jacobi theta-function. Using the representation \cite{GrRy.6}
\begin{equation}
    \frac{1}{\Ga\big((1-s)/2\big)}=\frac{i}{2\pi}\int_H d\omega e^{-\omega}(-\omega)^{(s-1)/2},
\end{equation}
where $H$ is the Hankel contour and the principal branch of the power function is taken, it is not difficult to obtain that
\begin{equation}
    I(\be,\tau)=\frac{i}{4\pi\tau}\int_H d\omega e^{-\omega}\ln\big[1+e^{-\be(-\omega/\tau)^{1/2}}\big].
\end{equation}
Developing the logarithm as a series and integrating it term by term, we arrive at
\begin{equation}
    I(\be,\tau)=\frac{\be}{8\pi^{1/2}\tau^{3/2}}\big[1-\vartheta_4(0,e^{-\be^2/(4\tau)}) \big]=\frac{\be}{8\pi^{1/2}\tau^{3/2}}+\frac1{4\tau}\sum_{n=-\infty}^\infty e^{-\tau \omega_n^2},
\end{equation}
where $\omega_n:=\pi(2n+1)/\beta$ are the Matsubara frequencies for fermions. Hence, the contribution of $\zeta_s\big((1-s)/2\big)$ to \eqref{HTE_Mellin} is
\begin{equation}\label{theta_repres}
    \int_0^\infty d\tau I(\be,\tau)\Tr e^{-\tau H^2_D}.
\end{equation}
This result is somewhat expected as the evaluation of the path-integral for fermions over $\psi^\dag$ and $\psi$ instead of $\bar{\psi}$ and $\psi$ gives at finite temperature (see, e.g., \cite{BerezPI,KapustB})
\begin{equation}
    \Ga^{(1)}_f=\frac12\sum_{n=-\infty}^\infty\ln\det(\omega_n^2+H_D^2),
\end{equation}
rather than the logarithm of the determinant of the square of the Dirac equation operator. We shall not use the representation \eqref{theta_repres} below, but it can be useful for the numerical evaluation of \eqref{HTE_Mellin}.

As we have already noted in Sec. \ref{Omega_Pot_T}, one can introduce the two definitions for the energy of vacuum fluctuations of the Dirac fermions: the $C$-symmetric \eqref{Evac1s} and the non $C$-symmetric \eqref{Evac_as0} one. In accordance with the non $C$-symmetric definition,
\begin{equation}\label{Evac_as}
    E_{vac}=-\sum_k E^{(-)}_k\underset{\be_0\rightarrow0}{\longleftarrow}2\frac{\partial}{\partial\be_0}\int_0^\infty d\omega\rho(-\omega)\ln(1+e^{-\be_0\omega})|_{\mu=0},
\end{equation}
where $\be_0$ is the regularization parameter and $E^{(-)}_k$ are the antiparticle energies. Using \eqref{zeta_part_v} and \eqref{HTE_ferm}, we obtain the nonrenormalized expression for the vacuum energy
\begin{multline}\label{Evac1}
    E_{vac}=\sum_{k=0}^{d-1}\frac{\Ga(D-k)\eta(D-k)}{(4\pi)^{d/2}}\be^{k-D}_0\Big[\frac{a'_{(k+1)/2}}{\Ga\big((D-k)/2\big)}-\frac{2 a_{k/2}}{\Ga\big((d-k)/2\big)} \Big] +\frac{\ln4}{\be_0}\frac{a'_{D/2}}{(4\pi)^{D/2}}-\\
    -\frac12\Big[\bar{\zeta}_s(-1/2)-\frac{a'_{1+d/2}}{2(4\pi)^{d/2}} +\frac{a_{D/2}}{(4\pi)^{D/2}}\ln\frac{4\be^2_0e^{2\ga}}{\pi^2} \Big],
\end{multline}
where one should set $\mu=0$. The first line contains the power divergencies only. The logarithmic divergence responsible for the conformal anomaly and the finite part are presented in the last line.

The $C$-symmetric vacuum energy $E^{c}_{vac}$ is defined in \eqref{Evac1s}. The explicit expression for $E^c_{vac}$ is obtained from \eqref{Evac1} by throwing out all the terms containing $a'_{k/2}$. The symmetric definition of the vacuum energy seems to be a more preferable one. It is symmetric with resect to the exchange of notions of a particle and an antiparticle. Besides, this definition implies that the finite part of the vacuum contribution is exactly canceled by the corresponding contribution from the thermal part of the $\Omega$-potential in the high-temperature limit and at zero chemical potential. We see from \eqref{Evac1} that, in the four-dimensional spacetime, the finite parts of the symmetric and nonsymmetric vacuum energies differ only by the surface term proportional to $a'_{5/2}$.

As it should be, the energy of vacuum fluctuations \eqref{Evac_as0}, \eqref{Evac1s} and the corresponding effective action are real-valued for the fields satisfying the conditions mentioned at the beginning of this section. The vacuum state is standardly defined (see, e.g., \cite{GrMuRaf} for details) as the ground state of the Hamiltonian of quantum Dirac fields. The particles are defined as the modes with $\omega_k>0$ (see \eqref{Sturm_Lioville_probl}), whereas the antiparticles are the modes with $\omega_k<0$. The total $\Omega$-potential with the contribution of the vacuum fluctuations becomes
\begin{equation}\label{one_loop_EA}
    L^{(1)}_{eff}=-\Omega^{(1)}_{tot}=-\Omega_{th}-E^{ren}_{vac}.
\end{equation}
The high-temperature expansion of $\Omega_{th}$ is presented in \eqref{HTE_ferm}, where the contribution from antiparticles should be taken into account, and $E^{ren}_{vac}$ is the renormalized vacuum energy (see \eqref{Evacren}). The term $\bar{\zeta}_s(-1/2)$ in $\Omega_{th}$ at zero chemical potential is exactly canceled by the analogous term in $E^{ren}_{vac}$ \cite{ElmSkag95,ElmSkag98,SkalDem,KazShip,gmse}. Hence, the total asymptotic expansion in $\be$ of the one-loop effective Lagrangian $L^{(1)}_{eff}$ induced by the Dirac fermions is expressed only in terms of the heat kernel expansion coefficients, when $\mu=0$. This property is an analog of the fact that the terms with $\s^l_\nu$, for $l$ even, vanish in the high-temperature expansion \eqref{HTE_gen}. Notice that, for a nonzero chemical potential, the cancelation of $\bar{\zeta}_s(-1/2)$ in $L^{(1)}_{eff}$ is not complete (see \cite{KazShip} and \eqref{omega_tot}).

The contributions of $\bar{\zeta}_s(-1/2)$, $\bar{\eta}_s(1/2-l)$, $l=\overline{0,\infty}$, in \eqref{HTE_ferm} are not expressed through the heat kernel expansion coefficients. Nevertheless, as is well known, the heat kernel expansion allows one to find the asymptotic expansion of these contributions in the inverse powers of a large mass $m$. Here we shall find such an expansion only for $\bar{\zeta}_s(-1/2)$ and $\bar{\eta}_s(1/2)$ appearing in the finite at $\be\rightarrow0$ part of the high-temperature expansion. Let
\begin{equation}\label{Hbar}
    \bar{H}^2:=(H_D-\e)^2-m^2,
\end{equation}
and
\begin{equation}\label{HKE2}
    \Tr e^{-\tau\bar{H}^2}\simeq\sum_{k=0}^\infty\frac{\tau^{(k-d)/2}}{(4\pi)^{d/2}}\bar{a}_{k/2}(\e).
\end{equation}
Then, substituting \eqref{Hbar} into this expression and developing it as a series in $\tau$ with the aid of \eqref{HKE}, we deduce
\begin{equation}\label{a_abar}
    a_{k/2}(\e)=\sum_{s=0}^{[k/2]}\frac{(-m^2)^s}{s!}\bar{a}_{k/2-s}(\e),\qquad \bar{a}_{k/2}(\e)=\sum_{s=0}^{[k/2]}\frac{m^{2s}}{s!}a_{k/2-s}(\e).
\end{equation}
Substituting \eqref{Hbar} into \eqref{zeta_nue} and using the expansion \eqref{HKE2}, we come to
\begin{equation}
\begin{split}
    \bar{\zeta}_s(-1/2)\simeq &\, \text{f.p.}\,\sum_{k=0}^\infty\frac{\Ga\big(\epsilon+(k-D)/2\big)}{\Ga(\epsilon-1/2)}(m^2)^{(D-k)/2-\epsilon}\frac{\bar{a}_{k/2}(\e)}{(4\pi)^{d/2}}=\\
    = &\, \text{f.p.}\,\sum_{s=0}^{[D/2]}\frac{\Ga(\epsilon-s)}{\Ga(\epsilon-1/2)}(m^2)^{s-\epsilon}\frac{\bar{a}_{D/2-s}(\e)}{(4\pi)^{d/2}} -\sideset{}{'}\sum_{k=0}^\infty \Ga\Big(\frac{k-D}2\Big) (m^2)^{(D-k)/2} \frac{\bar{a}_{k/2}(\e)}{(4\pi)^{D/2}}=\\
    = &\, \sum_{s=0}^{[D/2]}\Big[\ln\frac{m^2}{4e^{\ga-2}} -\psi(1+s)\Big] \frac{(-m^2)^s}{s!}\frac{\bar{a}_{D/2-s}(\e)}{(4\pi)^{D/2}} -\sideset{}{'}\sum_{k=0}^\infty \Ga\Big(\frac{k-D}2\Big) (m^2)^{(D-k)/2} \frac{\bar{a}_{k/2}(\e)}{(4\pi)^{D/2}}=\\
    = &\, \frac{a_{D/2}(\e)}{(4\pi)^{D/2}}\ln\frac{m^2}{4e^{\ga-2}} -\sum_{s=0}^{[D/2]}\frac{(-m^2)^s}{s!}\psi(1+s)\frac{\bar{a}_{D/2-s}(\e)}{(4\pi)^{D/2}} -\sideset{}{'}\sum_{k=0}^\infty \Ga\Big(\frac{k-D}2\Big) (m^2)^{(D-k)/2} \frac{\bar{a}_{k/2}(\e)}{(4\pi)^{D/2}},
\end{split}
\end{equation}
where $\text{f.p.}$ means the finite part with respect to $\epsilon\rightarrow0$, the relation \eqref{a_abar} has been used in the last line, and the prime at the sum sign recalls that all the terms that are singular due to the gamma function must be omitted. Employing the relation \cite{GrRy.6}
\begin{equation}
    \psi(1+n)=-\ga+H_n,\qquad H_n:=\sum_{k=1}^nk^{-1},\qquad H_0:=0,
\end{equation}
we have eventually
\begin{equation}\label{fin_part}
\begin{split}
    \bar{\zeta}_s(-1/2)\simeq \frac{a_{D/2}(\e)}{(4\pi)^{D/2}}\ln\frac{m^2e^2}{4} -\sum_{s=1}^{[D/2]}\frac{(-m^2)^s}{s!}H_s\frac{\bar{a}_{D/2-s}(\e)}{(4\pi)^{D/2}} -\sideset{}{'}\sum_{k=0}^\infty \Ga\Big(\frac{k-D}2\Big) (m^2)^{(D-k)/2} \frac{\bar{a}_{k/2}(\e)}{(4\pi)^{D/2}}.
\end{split}
\end{equation}
As far as $\bar{\eta}_s(1/2)$ is concerned, it is obtained from $\bar{\zeta}_s(-1/2)$ by means of \eqref{eta-zeta}.

Substituting \eqref{fin_part} into \eqref{Evac1}, we derive the asymptotic expansion
\begin{multline}\label{Evac2}
    E_{vac}\simeq\sum_{k=0}^{d-1}\frac{\Ga(D-k)\eta(D-k)}{(4\pi)^{d/2}}\be^{k-D}_0\Big[\frac{a'_{(k+1)/2}}{\Ga\big((D-k)/2\big)}-\frac{2 a_{k/2}}{\Ga\big((d-k)/2\big)} \Big] +\frac{\ln4}{\be_0}\frac{a'_{D/2}}{(4\pi)^{D/2}}-\\
    -\frac12\Big[\frac{a_{D/2}}{(4\pi)^{D/2}}\ln\frac{m^2\be^2_0e^{2\ga+2}}{\pi^2} -\frac{a'_{1+d/2}}{2(4\pi)^{d/2}} -\sum_{s=1}^{[D/2]}\frac{(-m^2)^s}{s!}H_s\frac{\bar{a}_{D/2-s}}{(4\pi)^{D/2}} 
    -\sideset{}{'}\sum_{k=0}^\infty \Ga\Big(\frac{k-D}2\Big) (m^2)^{(D-k)/2} \frac{\bar{a}_{k/2}}{(4\pi)^{D/2}}\Big],
\end{multline}
where $\mu=0$. The last term in this expression looks as the standard expansion of the one-loop effective action with respect to a large mass (see, e.g., \cite{BirDav.11,DeWGAQFT}). However, one should bear in mind that $\bar{a}_{k/2}$ depend on $m$. The maximal power of $m$ appearing in $\bar{a}_{k/2}$ will be found later in Sec. \ref{Lead_Term_HTE}.

The self-adjoint positive-definite Laplace type operator defining the heat kernel \eqref{zeta_nue} takes the form
\begin{equation}\label{Dirac_sq}
    (H_D-\e)^2=D_i^2+m^2+\frac12\s^{ij}F_{ij}-2P^2_0-2mP_0\ga^0,\qquad \s^{\mu\nu}:=\frac{i}{2}[\ga^\mu,\ga^\nu],
\end{equation}
where $P_0:=\e-A_0$ and $D_i=P_i-\al_i P_0$. The operator $i\ga^0\ga^5$ is self-adjoint in the Hilbert space specified by the condition \eqref{bound_conds}. Inserting $(i\ga^0\ga^5)^2$ under the trace sign in \eqref{zeta_nue}, we infer that
\begin{equation}
    \zeta_s(\nu;A_0)=\zeta_s(\nu;-A_0),\qquad \eta_s(\nu;A_0)=-\eta_s(\nu;-A_0),
\end{equation}
for $\e=0$. Remark that the derivative with respect to $\e$ in \eqref{eta-zeta} and in $a'_{k/2}$ at $\e=0$ can be replaced by the minus derivative with respect to the chemical potential $\mu$ entering into the definition of $A_0$. Therefore, we put henceforth $P_0=-A_0$ in \eqref{Dirac_sq}.

The operator $\ga^5$ is not compatible with the boundary condition \eqref{bound_conds} and does not map the Hilbert space into itself. Nevertheless, if we consider the spectral problem \eqref{Sturm_Lioville_probl} in the entire space, i.e., if we do not impose the boundary condition \eqref{bound_conds}, then $\ga^5$ becomes a self-adjoint operator in the respective Hilbert space. In this case, inserting $(\ga^5)^2$ under the trace sign in \eqref{zeta_nue}, we see that the spectral density $\rho(\omega)$ is symmetric with respect to the replacement $m\rightarrow-m$. In the next section, we shall explicitly show that the boundary condition \eqref{bound_conds} violates the symmetry of the spectral density under $m\rightarrow-m$.

\section{Leading terms of the high-temperature expansion}\label{Lead_Term_HTE}

Making use of the general formulas of the previous section, let us obtain the leading terms of the high-temperature expansion of the $\Omega$-potential of Dirac fermions in a bag. The first few coefficients of the heat kernel expansion associated with the Laplace type operator with nontrivial boundary conditions are given in \cite{Vassil,BGKV}. The higher coefficients without the surface terms are presented in \cite{Ven}. The requirement of self-adjointness of the Dirac Hamiltonian with the boundary condition \eqref{bound_conds} leads to
\begin{equation}\label{bound_conds_2}
    \Pi\ga^0(\hat{P}-m)\psi\Big|_{\Ga}=\ga^0(\nabla_n+m-\frac12 L_{aa})\bar{\Pi}\psi\Big|_{\Ga}=0,\qquad\bar{\Pi}=1-\Pi=\frac{1+i\hat{n}}{2},
\end{equation}
where $\hat{P}=\ga^\mu P_\mu$ and
\begin{equation}
    \nabla_i=\partial_i+iA_i-i\al_i A_0,\qquad \nabla_n=n_i\nabla_i,\qquad\al_i=-\al^i.
\end{equation}
We have introduced here the extrinsic curvature of the hypersurface $\Ga$ where the boundary condition is defined (see, for details, \cite{FursVass}):
\begin{equation}
    L_{ab}:=-e_a^\mu e_b^\nu\partial_\mu n_\nu+n_\la\Ga^\la_{\mu\nu}e_a^\mu e_b^\nu=-e_a^ie_b^j\partial_in_j,\quad a,b=\overline{1,2},
\end{equation}
where $e_a^\mu$ are orthonormal vectors tangent to the surface, $n_{\mu}=(0,\textbf{n})$ is the inward unit normal, and $\Ga^\la_{\mu\nu}$ is the Levi-Civita connection (in our case it is equal to zero). Hence, employing the notation from \cite{Vassil,BGKV}, we have from \eqref{Dirac_sq} and \eqref{bound_conds_2}:
\begin{equation}\label{ESOm}
\begin{gathered}
    \chi=i\hat{n},\qquad S=(m-\frac12 L_{aa})\bar{\Pi},\\
    E=2A_0^2-2mA_0\ga^0-\frac{1}{2}\s^{ij}F_{ij}-m^2,\qquad\Omega_{ij}=i(F_{ij}+\al_{[i}\partial_{j]}A_0-2\s_{ij}A_0^2).
\end{gathered}
\end{equation}
The expressions for the traces arising in evaluating the heat kernel expansion coefficients are given in Appendix \ref{Traces_Sec}. Substituting the traces found in \eqref{traces}-\eqref{traces3} to the general formulas\footnote{Notice that we did not check these formulas.} of \cite{Vassil,BGKV}, we obtain
\begin{equation}\label{ak}
\begin{split}
    a_0=&\,4\int d\spx,\qquad a_{1/2}=0,\qquad a_1=\int d\spx(8A_0^2-4m^2)+\int_\Ga d^2\tau\sqrt{h}[4m-\tfrac{2}{3}L_{aa}],\\
    a_{3/2}=&\,\frac{\sqrt{\pi}}{16}\int_\Ga d^2\tau\sqrt{h}[32m^2-16mL_{aa}-2L_{ab}L_{ab}+L_{aa}^2],\\
    a_{2}=&\int d\spx[2m^4+\tfrac23 F_{\mu\nu}F^{\mu\nu}+\tfrac43\partial_k\partial_k(A_0^2)]+\\
    &+\int_\Ga d^2\tau\sqrt{h}[\tfrac43\partial_nA_0^2 +16mA_0^2 -\tfrac43 m^3 -\tfrac23m^2L_{aa} +\tfrac2{15}mL_{aa}^2 -\tfrac{2}{5}mL_{ab}L_{ab}+\\
    &+\tfrac{17}{945}L_{aa}^3 +\tfrac{13}{315}L_{ab}L_{ab}L_{cc} -\tfrac{116}{945}L_{ab}L_{bc}L_{ac} -\tfrac1{15}L_{aa:bb}],\\
    a_{5/2}=&\,\sqrt{\pi}\int_\Ga d^2\tau\sqrt{h}[6m^2A_0^2-3mA_0^2L_{aa}+5m\partial_nA_0^2+\tfrac18\partial_nA_0^2L_{aa}+\tfrac18A_0^2L_{ab}L_{ab}-\tfrac1{16}A_0^2L_{aa}^2],
\end{split}
\end{equation}
where $\tau$ are the coordinates on the boundary surface, and $h$ is the determinant of the metric induced on $\Ga$. Notice that, in $a_2$, the last term in the first line is canceled by the first term in the second line. The last term in $a_2$ can be omitted. As for $a_{5/2}$, the terms that do not vanish under differentiation with respect to $\mu$ are only presented. The volume terms in the coefficients $a_{k}$, $k=\overline{0,2}$, coincide with the ``pseudo-trace'' expansion coefficients (\cite{FursVass}, problem 7.21, see also \eqref{pseudo_tr}). The derivatives with respect to $-\mu$ are written as
\begin{equation}
\begin{gathered}
    a'_0=a'_{1/2}=a'_{3/2}=0,\qquad a'_1=-16\int d\spx A_0,\qquad a'_2=-32\int_\Ga d^2\tau\sqrt{h}mA_0,\\
    a'_{5/2}=\sqrt{\pi}\int_\Ga d^2\tau\sqrt{h}[6mA_0L_{aa} -12m^2A_0 -10m\partial_nA_0 -\tfrac14\partial_nA_0L_{aa} -\tfrac14A_0L_{ab}L_{ab} +\tfrac18A_0L_{aa}^2].
\end{gathered}
\end{equation}
As we have discussed above, the last expression gives the difference between the finite parts of the $C$-symmetric and the non $C$-symmetric vacuum energies of Dirac fermions. Upon renormalization, this contribution has to be completely canceled out by the counterterms since the integrand in $a'_{5/2}$ has the dimension not exceeding $3$ and is not Lorentz-invariant, even if one restores the dependence on the normal vector $n^\mu$. Thus, after renormalization, the $C$-symmetric and the non $C$-symmetric vacuum energies coincide.

The leading correction to the high-temperature expansion due to the nontrivial boundary condition \eqref{bound_conds} comes from the coefficient $a_1$. When the boundary $\Ga$ is a sphere of the radius $R$, the extrinsic curvature $L_{ab}=\de_{ab}/R$. In this case, substituting \eqref{ak} into \eqref{HTE_ferm} doubled and setting $m=0$, we reproduce the leading correction in (48) of \cite{DeFranc} to the free energy due to the boundary condition \eqref{bound_conds}. As we see, the coefficient $a_1$ and the corresponding contribution to \eqref{HTE_ferm} are not symmetric with respect to $m\rightarrow-m$. Consequently, the spectral density $\rho(\omega)$ is not symmetric under this replacement either. The volume contributions to $a_{k/2}$ are invariant with respect to $m\rightarrow-m$.

The coefficient $a_2$ is prefixed to $\ln\be_0$ in the vacuum energy \eqref{Evac1} and, as is well known, related to the conformal anomaly. Formula \eqref{ak} provides the generalization of the standard expression for $a_2$ to the case of the electrically charged Dirac fields obeying the boundary condition \eqref{bound_conds}. In the case of the neutral massive Dirac field confined to a sphere, the surface term coming from $a_2$ coincides with (3.1) of \cite{EliBorKir}. The surface term in $a_2$ is not symmetric under $m\rightarrow-m$. The term $mA_0^2$ in the surface contribution to $a_2$ should be completely canceled out by the counterterm inasmuch as it is not Lorentz-invariant.

In order to renormalize the other contributions to the vacuum energy, we need to find the terms in $\bar{a}_k$ with the mass dimension less than or equal to $4$, the dimension of the coefficients depending on $m$ being not taken into account. So we have to determine the maximal power of $m$ entering into $\bar{a}_{k/2}$. As for the volume terms in $\bar{a}_{k/2}$, such contributions come from the second term in $E$. The maximal power of $m$ in the surface terms stems from the first term in $S$. In this case,
\begin{equation}
    \bar{a}_{k/2}|_{\text{v.t.}}\sim (mA_0)^{k/2},\qquad \bar{a}_{k/2}|_{\text{s.t.}}\sim m^{k-1}.
\end{equation}
Therefore, if one takes into account the surface terms, the expansion in $m^{-2}$ in the last contribution in \eqref{fin_part} and \eqref{Evac2} is, in fact, not an expansion in the rising powers of $m^{-1}$. In order to obtain the large mass expansion with the surface contributions, one has to resum the expansion in \eqref{fin_part} and \eqref{Evac2}. We leave this issue for a future research and, in considering the large mass expansion, shall cast out the surface terms. The expansion obtained thereby is valid for the problem \eqref{Sturm_Lioville_probl} posed in the whole space without the boundary condition \eqref{bound_conds}.

The necessary for the renormalization procedure terms are
\begin{equation}
\begin{split}
    \bar{a}_{3}|_{\text{v.t.}}=&\,\tfrac1{360}\int d\spx\tr(-30E_{;i}E_{;i}+60 E^3)=-\tfrac43\int d\spx m^2\mathbf{E}^2,\\
    \bar{a}_{4}|_{\text{v.t.}}=&\,\tfrac1{4!}\int d\spx\tr E^4=\tfrac{8}3\int d\spx m^4A_0^4.
\end{split}
\end{equation}
The other volume terms in $\bar{a}_{k/2}$, $k\geq3$, are suppressed in \eqref{fin_part}, \eqref{Evac2} by a power of a mass. The asymptotic large mass expansion of the vacuum energy \eqref{Evac2} without the surface terms takes the form
\begin{multline}\label{Evac_lme}
    E_{vac}=-\int d\spx\Big[\frac{7\pi^2}{60\be_0^4} +\frac{6\zeta(3)}{\be_0^{3}}\tilde{A}_0 +\frac{1}{6\be_0^2}(\tilde{A}_0^2-\tfrac12m^2)-\\
    -\Big(\frac{m^4}{16\pi^2}+\frac{F_{\mu\nu}F^{\mu\nu}}{48\pi^2} \Big)\ln\frac{m^2\be_0^2e^{2\ga+2}}{\pi^2} -\frac{m^2\tilde{A}_0^2}{4\pi^2} +\frac{\tilde{A}_0^4}{12\pi^2} +\frac{3m^4}{32\pi^2} -\frac{\mathbf{E}^2}{24\pi^2} +O(m^{-2})\Big],
\end{multline}
where $\tilde{A}_0:=A_0|_{\mu=0}$. In accordance with the standard renormalization rules (see, e.g, \cite{BogShir}), all the written terms in \eqref{Evac_lme} should be completely canceled out by the counterterms provided one normalizes the effective action at a zero photon momentum and a zero chemical potential. As a result, we have from \eqref{Evac1}:
\begin{equation}\label{Evacren}
    E^{ren}_{vac}=E_{vac}+\text{c.t.}=\frac{a_2}{32\pi^2}\ln(m^2e^2)-\frac12\bar{\zeta}_s(-1/2)|_{\mu=0} +\int d\spx\Big[\frac{m^2\tilde{A}_0^2}{4\pi^2}-\frac{\tilde{A}_0^4}{12\pi^2}-\frac{3m^4}{32\pi^2}+\frac{\mathbf{E}^2}{24\pi^2}\Big].
\end{equation}
Substituting this expression into \eqref{one_loop_EA}, we find that the complete asymptotic high-temperature expansion of the total one-loop $\Omega$-potential of the Dirac fermions with  account for the contribution of antiparticles and the vacuum energy (but without the surface terms) reads as
\begin{multline}\label{omega_tot}
    -\Omega^{(1)}_{tot}\simeq\int d\spx\Big[\frac{7\pi^2}{180\be^4} +\frac{2A_0^2-m^2}{12\be^2} -\frac{m^2\tilde{A}_0^2}{4\pi^2} +\frac{\tilde{A}_0^4}{12\pi^2} -\frac{\mathbf{E}^2}{24\pi^2} +\frac{3m^4}{32\pi^2}\Big]-\\
    -\frac12\big[\bar{\zeta}_s(-1/2) -\bar{\zeta}_s(-1/2)|_{\mu=0}\big]
    -\frac{a_2}{32\pi^2}\ln\frac{4m^2\be^2e^{2\ga}}{\pi^2}
    +\sum_{l=1}^\infty \frac{\Ga(-2l)\eta(-2l)}{(4\pi)^{d/2}}\frac{a_{l+2}\be^{2l}}{\Ga(1/2-l)} .
\end{multline}
Let us note once again that, at $\mu=0$, this expansion is expressed solely in terms of the heat kernel expansion coefficients $a_k$, and this is not a large mass expansion. Because of the logarithmic contribution, this expression is not defined when $m\rightarrow0$, i.e., there exists an infrared divergence in this limit, at least at zero chemical potential. One can try to solve this problem by introducing an effective mass for the fermions and sum thereby an infinite number of the ring diagrams (see, e.g., \cite{Miransky,KapustB}). The contribution of the photons to \eqref{omega_tot} with their effective mass should be also taken into account.

The leading terms of the high-temperature expansion of the $\Omega$-potential \eqref{HTE_ferm} expanded in a large mass $m$ without the contribution of antiparticles and the surface terms are written as
\begin{multline}\label{HTE_LME}
    -\Omega=\int d\spx\Big[\frac{7\pi^2}{360\be^4} -\frac{3\zeta(3)}{2\pi^2\be^3}A_0 +\frac{2A_0^2-m^2}{24\be^2} +\frac{\ln2}{2\pi^2\be} (m^2A_0 -\tfrac23 A_0^3)-\\
    -\Big(\frac{m^4}{32\pi^2} +\frac{F_{\mu\nu} F^{\mu\nu}}{96\pi^2}\Big)\ln\frac{m^2\be^2e^{2\ga}}{\pi^2} -\frac{m^2A_0^2}{8\pi^2} +\frac{A_0^4}{24\pi^2} -\frac{\mathbf{E}^2}{48\pi^2} +\frac{3m^4}{64\pi^2} \Big]+\cdots
\end{multline}
For $\mathbf{E}=0$, this expression coincides with the leading terms of the expansion (105) of \cite{KalKaz3} provided one changes the sign of the chemical potential. If we take into account the contribution of antiparticles, i.e., if we take the doubled symmetric in $A_0$ part of \eqref{HTE_LME}, then we obtain formula (4.2) of \cite{ElmSkag98} up to a renormalization ambiguity. The nonsymmetric in $A_0$ part of \eqref{HTE_LME} can be used, for example, to find the number of electron-positron pairs in the system (see, for details, \cite{KhalilB,KalKaz3}) at finite temperature.

\section{Comparison with other approaches}\label{Comparison}

Let us compare the method we have developed with the approaches presented in the literature. Of course, we do not pretend here to a complete description of the whole literature devoted to this subject, but mention only the main approaches that we know and that are close to the method used by us.

Usually, the one-loop correction to the effective action is expressed in terms of the trace \cite{Schwing.10,DeWGAQFT,BirDav.11}
\begin{equation}\label{Trexp}
    \Tr_4 e^{-\tau \mathcal{K}},
\end{equation}
where
\begin{equation}\label{Dirac_sq2}
    -\mathcal{K} =-\hat{P}^2+m^2=-(i\partial_t-A_0)^2+P_i^2+m^2+\frac12\s^{\mu\nu}F_{\mu\nu}.
\end{equation}
The quantity \eqref{Trexp} cannot be directly used for a nonperturbative evaluation of the effective action since the operator under the trace sign is not trace-class \cite{KazMil1,KazMil2}. Moreover, any operator function $f(\mathcal{K})$ different from zero is not trace-class. In order to see this, it is sufficient to consider the trace in the momentum representation in the domain of large $|p_\mu|$. Then
\begin{equation}\label{Spfk}
    \Tr_4 f(\mathcal{K})\approx\int \frac{d^4xd^4p}{(2\pi)^4} f(\mathcal{K}(p,x)),
\end{equation}
where $\mathcal{K}(p,x)$ is the principal symbol of $\mathcal{K}$, which is obtained from $\mathcal{K}$ by the replacement of $i\partial_\mu$ by $p_\mu$ and neglecting the terms of the lower order in $p_\mu$. The integrand of \eqref{Spfk} is invariant under the Lorentz transformations of the vector $p_\mu$ preserving the spacetime metric $g_{\mu\nu}(x)$ at the fixed point $x$. The action of the Lorentz group foliates the momentum space onto the noncompact orbits where the integrand is constant. Therefore, the integral \eqref{Spfk} is not absolutely convergent and can converge only conditionally what implies that the operator $f(\mathcal{K})$ is not trace-class. Perturbatively, this property is revealed as the conditional convergence of the renormalized Feynman diagrams \cite{Dyson49,Zimmer} in the Minkowski spacetime. For the electromagnetic field of a general configuration, the expression \eqref{Trexp} seems to be understood only perturbatively, and the trace should be evaluated in a certain basis (see, for example, \eqref{pseudoheat_kern}).

In the case of the stationary electromagnetic potential $A_\mu$, one can consider the operator
\begin{equation}\label{Dirac_sq3}
    -\mathcal{K} (\omega)=-\hat{P}^2+m^2=-(\omega-A_0)^2+P_i^2+m^2+\frac12\s^{\mu\nu}F_{\mu\nu},
\end{equation}
and pose the so-called nonlinear spectral problem for it (see the book \cite{FursVass} and references therein). Then, formally, the evaluation of the effective action at finite and zero temperatures is reduced to the evaluation of the expression (see Secs. \ref{HTE_Deriv}, \ref{Omega_Pot_T} above)
\begin{equation}\label{Trexp2}
    \Tr e^{-\tau \mathcal{K}(\omega)}.
\end{equation}
The operator $\mathcal{K}(\omega)$ is of Laplace type, the corresponding heat kernel is a trace-class operator, and \eqref{Trexp2} is uniquely defined (see, e.g., \cite{Shubin,Agran,Gilkey2.11}). However, in the presence of the external electric field, the operator $\mathcal{K}(\omega)$ is not Hermitian owing to the last term in \eqref{Dirac_sq3} (see, e.g., \cite{DunHall,DeWGAQFT,BagGit1}). Therefore, the method of Sec. \ref{HTE_Deriv} is not immediately applicable.

In \cite{FursVass}, Sec. 6, the pseudo-trace,
\begin{equation}\label{pseudo_tr}
    K(t)=\frac12\sum_ke^{-t\omega_k^2},
\end{equation}
is introduced. Here $\omega_k$ are the solutions to the equation
\begin{equation}\label{spectr_eqn}
    \omega^2=\la_k(\omega),
\end{equation}
and $\la_k(\omega)$ are the eigenvalues of the operator $\omega^2-\mathcal{K}(\omega)$. The high-temperature expansion of the $\Omega$-potential is expressed through the expansion of
\eqref{pseudo_tr} for small $t$. As long as the operator \eqref{Dirac_sq3} is not self-adjoint, equation \eqref{spectr_eqn} can have complex roots. The methods presented in Sec. 6 and  problem 7.21 of \cite{FursVass} (see also \cite{Furshep}) are not applicable, at least immediately, in this case since they are essentially based on the self-adjointness of the operator $\omega^2-\mathcal{K}(\omega)$, i.e., on the real-valuedness of $\la_k(\omega)$ for real $\omega$ (see the remark after (6.6) of \cite{FursVass} and the requirements ii), iii) after (2.7) of \cite{Furshep}). Besides, the physical interpretation of the complex solutions of the nonlinear spectral problem for the operator \eqref{Dirac_sq3} is unclear. The kernel of the operator \eqref{Dirac_sq3} can be larger than the set of stationary solutions of the initial Dirac equation.

Unfortunately, this issue is poorly discussed in the literature. In \cite{Schwing.10}, this problem is not touched at all. In the standard textbook \cite{GrMuRaf}, it is wrongly assumed that $\mathcal{K}$ is a self-adjoint operator (see the remark after (10.108b)). In \cite{GFSh.3}, Sec. 6.2.3, it is supposed that $\mathcal{K}$ possesses a real-valued spectrum and a set of eigenvectors that is complete and orthonormal with respect to the Lorentz-invariant scalar product. In \cite{KabShab}, it is assumed, in considering the thermodynamic properties of fermions in external electromagnetic fields, that the electric field is completely screened, and actually the case of $\textbf{E}=0$ is investigated. A more detailed discussion of this issue is presented in \cite{DeWGAQFT}, p. 551, but its solution is not given there. Besides, it is supposed that the imaginary parts of the eigenvalues of the operator $\mathcal{K}$ are of the same sign. However, in the case $\mathbf{A}=0$, it is easy to show using the Majorana representation of $\ga$-matrices (see, e.g., \cite{BagGit1}) that if $\e_k$ and $\psi_k$ are the eigenvalue and the eigenvector of \eqref{Dirac_sq3}, then $\e^*_k$ and $\ga^0\psi^*_k$ are also the eigenvalue and the eigenvector of \eqref{Dirac_sq3} (we do not impose the boundary condition \eqref{bound_conds}). In this representation, equation \eqref{Dirac_sq3} with the potential $A_0(x_2)$ becomes
\begin{equation}
    [(\omega-A_0)^2+\Delta-m^2-iE_2]u_k=\e_k u_k,\qquad [(\omega-A_0)^2+\Delta-m^2+iE_2]\ups_k=\e_k \ups_k,\qquad
    \psi_k=\left[
           \begin{array}{c}
             u_k \\
             \ups_k \\
           \end{array}
         \right],
\end{equation}
where $E_2=-\partial_2A_0$. We see that, for the function $A_0(x_2)$ of a general form, the operator \eqref{Dirac_sq3} possesses the complex eigenvalues $\e_k$. Indeed, let $\e_k\in\mathbb{R}$ for some $A_0(x_2)$. Then, perturbing the potential by some $\de A_0(x_2)$, we obtain with the help of the standard formula of perturbation theory that $\de\e_k\in \mathbb{C}$.

Nevertheless, the methods based on the use of the expression \eqref{Trexp2} are nonperturbatively applicable for some classes of the electromagnetic fields. For example, such electromagnetic fields are
\begin{enumerate}
  \item The fields satisfying the condition $F_{\mu\nu}n^\nu=0$, where $n^\mu$ is a constant isotropic vector. The solution of the Dirac equation with these fields is reduced to the solution of the KG type equation with a Hermitian operator (see, for details, \cite{BagGit1});
  \item The fields with $A_0=0$. In this case, the operator \eqref{Dirac_sq3} is self-adjoint, and \eqref{Dirac_sq} is reduced to $\omega^2-\mathcal{K}(\omega)$ when $\e=0$. Therefore, the approaches expounded in Secs. \ref{HTE_Deriv} and \ref{HTE_Dir} are equivalent. The vacuum energy and the $\Omega$-potential written in terms of the spectral zeta-function associated with the squared Dirac equation operator are presented in this case in many papers and books (see, e.g., \cite{KhalilB,Drozd,Dunne2,AndNayTran,KalKaz3} and references therein);
  \item The constant electromagnetic field. In this case, the last term in \eqref{Dirac_sq3} can be factored out from the exponent in \eqref{Trexp2}, and the problem is reduced to the respective problem for the KG equation with a Hermitian operator.
\end{enumerate}

There is also a third approach to the problem based on the Wick rotation prescription (see, e.g., \cite{DesGrigSem,McKeShub,BytsZerb1,BytsZerb2,ByVaZe,KurVass} for details). According to this method, the one-loop correction to the effective action induced by Dirac fermions formally reads as
\begin{equation}\label{Dirac_eucl}
    \ln\det(\hat{P}_E-im)=\ln\det(\hat{P}_E+im)=\frac12\ln\det(\hat{P}^2_E+m^2),
\end{equation}
where $\hat{P}_E=\ga^\mu_EP_\mu$ and
\begin{equation}
    \{\ga^\mu_E,\ga^\nu_E\}=2\de^{\mu\nu},\qquad(\ga^\mu_E)^\dag=\ga^\mu_E.
\end{equation}
The gauge connection $A^E_\mu$ included into $P_\mu$ is real. The chemical potential enters \eqref{Dirac_eucl} as the imaginary additive $i\mu$ to $A^E_0$. The spectrum of $\hat{P}_E$ is real-valued for $\mu=0$ and, in the even-dimensional space, is symmetric with respect to zero. This follows from the fact that there exists the matrix $\ga^5$ in the even-dimensional space that maps the eigenvectors associated with the positive eigenvalues to the eigenvectors associated with the negative eigenvalues. It is assumed in \eqref{Dirac_eucl} that the determinant is regularized in such a way that this symmetry of the spectrum of $\hat{P}_E$ is preserved. The determinant on the right-hand side of \eqref{Dirac_eucl} can be evaluated, for example, with the help of the spectral zeta-function. Then, in the even-dimensional space, the multiplicative anomaly arises (see, for details, \cite{KontsVish,BytsZerb2}) which leads to the additional corrections to the expression \eqref{Dirac_eucl}. Unfortunately, in the case $A_0\neq0$, we do not know a rigorous proof of the equivalence (in the nonperturbative sense) of this approach to the Minkowski spacetime methods presented, for example, in Secs. \ref{HTE_Deriv}, \ref{HTE_Dir} above, or in \cite{FursVass}. For $A_0=0$, it was shown in \cite{BytsZerb2} that these approaches are equivalent for ultrastatic spacetimes. Notice that the Euclidean approach developed in \cite{FursVass} is distinct from that is described here. In \cite{FursVass}, having performed the Wick rotation, the background fields become complex (see Sec. 7.4, \cite{FursVass}), and the corresponding wave operator becomes non self-adjoint even in the massless case.

All the methods mentioned above must be perturbatively equivalent to the approach developed in Sec. \ref{HTE_Dir} even in the case $A_0\neq0$. In other words, the asymptotic expansions of the effective action in the coupling constant or, what is the same in the case at issue, in the inverse powers of a large mass must coincide up to a renormalization ambiguity. This is fulfilled, in particular, if $a_{k/2}$ coincide with the expansion coefficients of the pseudo-trace \eqref{pseudo_tr} associated with the operator \eqref{Dirac_sq3} (see \eqref{Evac2}). In Sec. \ref{Lead_Term_HTE}, we showed that this relation holds for $k=0,2,4$ without the surface terms. However, we have not succeeded in proving this relation for an arbitrary $k$. The perturbative equivalence to the Euclidean approach entails that the coefficients of the large mass expansion in the last line of \eqref{Evac2} multiplied by $T$, where $T$ is the observation period, are equal to the heat kernel expansion coefficients constructed for the operator $\hat{P}^2_E$, where the inverse Wick rotation and renormalization are performed. The direct proof of this statement deserves a separate study. Let us stress that, at first, $\bar{a}_{k/2}$ in \eqref{Evac2} must be expanded in $m^{-1}$.

Notice that the method of deriving the high-temperature expansion that employs the squaring of the Dirac Hamiltonian was used in \cite{Furs1,Furs2} in the case of a stationary gravitational background. The background connection associated with the internal gauge symmetries was not considered in \cite{Furs1,Furs2}. Besides, in \cite{Furs1,Furs2,FursVass}, the expression for the complete high-temperature expansion was not derived. The method we have developed in the previous sections allows us to obtain the complete high-temperature expansion of the contributions of particles and, separately, of antiparticles to the $\Omega$-potential, the nontrivial boundary conditions being also taken into account. This expansion is expressed in terms of the spectral zeta-function of the positive-definite self-adjoint operator of Laplace type. In \cite{Furs1,Furs2,FursVass,Furshep}, the contributions from the boundary conditions and the separate contributions from particles and antiparticles were not considered. In (1.66) of \cite{BytsZerb2}, the high-temperature expansion was obtained in terms of pseudo-differential operators in the case $A_0=0$.

In \cite{ElmSkag95,ElmSkag98}, the approximate expression for the one-loop thermodynamic potential of the Dirac particles was found in a stationary external electromagnetic field. In this expression, the zeroth and the first derivatives of $A_\mu$ are only taken into account. Several leading terms of the high-temperature expansion were also explicitly found in \cite{ElmSkag95,ElmSkag98}. Unfortunately, the rigorous nonperturbative derivation of formula (2.1) in \cite{ElmSkag98} is absent. In \cite{ElmSkag95}, the reference is made to the paper \cite{Schwing.10}, so the approach of \cite{ElmSkag95,ElmSkag98} appears to be equivalent to the method based on the expression \eqref{Trexp2} that we have discussed above. The expression  (4.2) of \cite{ElmSkag98} for the first leading terms of the high-temperature expansion coincides with \eqref{HTE_LME} up to renormalization provided that the doubled symmetric in $A_0$ part of \eqref{HTE_LME} is taken. In \cite{GusShov,Shovk99,Shovkovy} (see also \cite{LoewRoj}), a systematic method of deriving the gradient corrections to the effective action both at zero and finite temperatures was elaborated, the Schwinger's answer for the constant electromagnetic field generalized to a finite temperature being taken as the zeroth order approximation. According to this approach, the vacuum energy contains an imaginary contribution when the electric field is present. It is not clear whether the expression (1) of \cite{Shovk99} is well-defined nonperturbatively or not.

\section{Conclusion}

To sum up, we have obtained an explicit well-defined nonperturbative expression for the complete high-temperature expansion of the one-loop $\Omega$-potential induced by scalars and Dirac fermions with nontrivial boundary conditions. The contributions of particles and antiparticles were treated separately that allows one to find the number of pairs of particles and antiparticles in the system (see, for details, \cite{KhalilB,KalKaz3}). The high-temperature expansion can be used to obtain the one-loop energy of vacuum fluctuations at zero temperature both for bosons and fermions. So we have obtained the vacuum energies for bosons \cite{KalKaz1,KalKaz2,KalKaz3} and fermions. In particular, we have found the explicit expression for the difference between the $C$-symmetric and the non $C$-symmetric vacuum energies (see, for details, \cite{GrMuRaf}) for the Dirac fermions. It turns out that the finite parts of these energies differ only by a surface term, which is absent if the Dirac fields obey the standard asymptotic conditions at spatial infinity, i.e., the condition \eqref{bound_conds} is not imposed. Thus, in this case, we may conclude that both definitions of the vacuum energy give the same result upon renormalization of the divergent part. In a general case, we have shown that this finite surface term should be completely canceled out by the respective finite counterterm. We have also found the leading correction to the high-temperature expansion due to the nontrivial boundary condition \eqref{bound_conds}. Besides, the correction to the coefficient $a_2$ responsible for the conformal anomaly has been derived in the case of the massive charged Dirac fields obeying the MIT bag boundary condition \eqref{bound_conds}. Finally, we have proved that the complete asymptotic in $\be$ high-temperature expansion of the one-loop effective action induced by the Dirac fermions at zero chemical potential is expressed solely in terms of the heat kernel expansion coefficients even in the case when the mass of the field is not large.

The difference between the $C$-symmetric and non $C$-symmetric definitions of the vacuum energies can be found for quarks in a bag. However, at first, one should generalize the above considerations to the case of non-Abelian background gauge fields. Another possibility to examine experimentally these two definitions can be realized in a graphene sheet with boundaries. In this case, the non $C$-symmetric definition is even more preferable from the physical viewpoint as the ``Dirac sea'' (the electrons in the valence band) is really present. The Lorentz symmetry is not a fundamental symmetry in this case either. Therefore, one may impose different normalization conditions than that we have used in Sec. \ref{Lead_Term_HTE}.

Several questions remain unanswered and are left for a future research. First, we did not prove by an explicit calculation that, up to a renormalization ambiguity, the large mass expansion of the vacuum energy \eqref{Evac2} coincides with the standard answer given, for example, by the Euclidean approach. Second, we did not obtain the large mass expansion of the vacuum energy with the surface contributions taken into account. The naive heat kernel expansion is to be resummed in this case. The third problem, which seems to be the most difficult one, is to prove (or disprove) the nonperturbative equivalence of the different approaches to the evaluation of the one-loop effective action at zero and finite temperatures. Of course, the procedure developed admits immediate generalizations to the curved spacetime, the chiral gauge connections, and other space dimensions. We leave these generalizations for a separate study.

\appendix
\section{Some analytic properties of the Mellin transform}\label{Analyt_Prop}

Let us consider the analytic properties of the integrals of the form
\begin{equation}\label{Ila}
    I_\la=\int_0^\La dx x^\la\vf(x),\quad 0<\La<+\infty,
\end{equation}
as functions of the complex variable $\la$. The following theorem holds for such integrals (see, e.g., \cite{GSh,Wong}).

\begin{thm}\label{anal_int_prop}
  Let $\vf(x)$ be absolutely integrable on $(0,\La]$, and for $x\rightarrow+0$ the following asymptotic expansion takes place
  \begin{equation}\label{asympt_exp}
    \vf(x)=\sum_{k=0}^N a_k x^k+O(x^{N+1}).
  \end{equation}
  Then
  \begin{enumerate}
    \item The function $I_\la$ is analytic for $\re\la>-1$ and can be continued analytically to the domain $\re\la>-2-N$, where it possesses the simple poles at the points $\la=-k$, $k=\overline{1,N+1}$, with the residues $a_{k-1}$, respectively;
    \item $I_\la\rightarrow0$ for $|\im\la|\rightarrow\infty$ in the domain $\re\la>-2-N$.
  \end{enumerate}
\end{thm}
\begin{proof}
  The integrals $I_\la$ and $\partial_\la I_\la$ are absolutely convergent for $\re\la>-1$, and, consequently, $I_\la$ is an analytic function of $\la$ in this domain. In this domain, we have
  \begin{equation}
    I_\la=\int_0^\La dx x^\la[\vf(x)-a_0]+a_0\int_0^\La dx x^\la=\int_0^\La dx x^\la[\vf(x)-a_0]+a_0\frac{\La^{\la+1}}{\la+1}.
  \end{equation}
  The right-hand side of this equality is defined for $\re\la>-2$, $\la\neq-1$, and provides an analytic continuation of $I_\la$ to this domain. The function $I_\la$ has a simple pole at the point $\la=-1$ with the residue $a_0$. Continuing this process by subtracting the terms of the asymptotic series \eqref{asympt_exp} from $\vf(x)$, we arrive at the first assertion of the theorem.

  For $\re\la>-2-N$, we have
  \begin{equation}\label{Ila_analyt}
    I_\la=\int_0^\La dx x^\la\Big[\vf(x)-\sum_{k=0}^Na_kx^k\Big]+\sum_{k=0}^Na_k\frac{\La^{\la+k+1}}{\la+k+1}.
  \end{equation}
  Making the substitution $y=\ln x$, we reduce the integral in this expression to the Fourier transform. In the domain of the $\la$ plane under consideration, the integrand is absolutely integrable, and, consequently, according to the Riemann-Lebesgue lemma (see, e.g., \cite{ReedSimon2}), the integral in \eqref{Ila_analyt} tends to zero for $|\im\la|\rightarrow\infty$. The out of the integral terms in \eqref{Ila_analyt} also tend to zero when $|\im\la|\rightarrow\infty$ and $\La>0$.
\end{proof}

In many cases, this theorem allows one to investigate the analytic properties of the function $I_\la$ without an explicit evaluation of the integral \eqref{Ila}. If the upper limit of the integral \eqref{Ila} is $\La=+\infty$, then one can split the integral into two ones and reduce them to \eqref{Ila} by a change of the integration variable, $x\rightarrow x^{-1}$, in the integral with the infinite integration limit. If the asymptotic expansion \eqref{asympt_exp} is performed in the fractional powers of $x$, then one can bring the integral at hand to the form \eqref{Ila} by a suitable change of the integration variable $x$ and redefinition of $\la$. Differentiating $I_\la$ with respect to $\la$, one can generalize the theorem to the case when the asymptotic expansion of $\vf(x)$ contains the terms $x^k\ln^l x$. In this case, $I_\la$ possesses multiple poles in the complex $\la$ plane. Theorem \ref{anal_int_prop} applied to the spectral density of the self-adjoint operator $A>0$ gives the structure of singularities of the spectral zeta-function $\zeta(\nu,A)$ in the complex $\nu$ plane and relates the residues of the spectral zeta-function to the heat kernel expansion coefficients.

\section{Traces}\label{Traces_Sec}

In this appendix, we present the traces appearing in the calculation of the heat kernel expansion coefficients:
\begin{equation}\label{traces}
\begin{gathered}
    \tr S=2m-L_{aa},\qquad \tr S^2=2(m-\frac12 L_{aa})^2,\qquad \tr S^3=2(m-\frac12 L_{aa})^3,\\
    \tr E=4(2A_0^2-m^2),\qquad\tr(\chi E)=0,\qquad \tr(SE)=(2m-L_{aa})(2A_0^2-m^2), \\
    \tr E^2=4m^4+16A_0^4+2F_{ij}F_{ij},\qquad \tr\nabla_k\nabla_kE=8\partial_k\partial_k A_0^2,\qquad \tr\Omega_{ij}\Omega_{ij}=-4[F_{ij}F_{ij}+4(\partial_iA_0)^2+24 A_0^4],\\
    \tr(\chi_{:a}\chi_{:a})=4\partial_a n_i\partial_a n_i=4L_{ab}L_{ab},\qquad\tr S_{:aa}=-L_{aa:bb},\qquad\tr(\chi\chi_{:a}\Omega_{an})=8A_0^2L_{aa},\\
    \tr(\chi_{:a}\chi_{:b}L_{ab})=4L_{ab}L_{ac}L_{bc},\qquad\tr(\chi_{:a}\chi_{:a}S)=(2m-L_{cc})L_{ab}L_{ab},\\
    \tr[(240\bar{\Pi}-120\Pi)E_{;n}]=480\partial_nA_0^2+2880mA_0^2,
\end{gathered}
\end{equation}
where $\partial_a:=e^i_a\partial_i$, $\partial_n=n_i\partial_i$, and the colon denotes the covariant derivative on the boundary $\Ga$. In calculating the coefficient $a_{5/2}$, the following traces appear (see \cite{BGKV}). In $\mathcal{A}^1_5$:
\begin{equation}\label{traces1}
\begin{gathered}
    \tr(\chi E_{;nn})=24m\partial_n A_0^2,\qquad \tr(E_{;n}S)=(4m-2L_{aa})(\partial_n (A_0^2) +2mA_0^2),\qquad \tr(\chi E^2)=0,\\
    \tr(\chi E_{:aa})=16mA_0^2L_{aa},\qquad \tr(SS_{:aa})=0,\qquad \tr(ES^2)=A_0^2(2m-L_{aa})^2.
\end{gathered}
\end{equation}
In $\mathcal{A}^2_5$:
\begin{equation}\label{traces2}
\begin{gathered}
    \tr[(90\bar{\Pi}+450\Pi)L_{aa}E_{;n}]=720L_{aa}(3\partial_n(A_0^2)-4mA_0^2),\qquad\tr(L_{aa}S_{:bb})=0,\qquad\tr(L_{ab}S_{:ab})=0,\\
    \tr(L_{aa}SE)=2A_0^2L_{aa}(2m-L_{bb}),\qquad\tr[(195\bar{\Pi}-105\Pi)L_{aa}^2E]=360A_0^2L_{aa}^2,\\
    \tr[(30\bar{\Pi}+150\Pi)L_{ab}L_{ab}E]=720A_0^2L_{ab}L_{ab}.
\end{gathered}
\end{equation}
In $\mathcal{A}^3_5$:
\begin{equation}\label{traces3}
\begin{gathered}
    \tr(E^2)=16A_0^4,\qquad \tr(\chi E\chi E)=16 A_0^2(A_0^2-2m^2),\qquad\tr(\Omega_{ab}\Omega_{ab})=-96A_0^4,\\
    \tr(\chi\Omega_{ab}\Omega_{ab})=0,\qquad\tr(\chi\Omega_{ab}\chi\Omega_{ab})=-96A_0^4,\qquad\tr(\Omega_{an}\Omega_{an})=-48 A_0^4,\\
    \tr(\chi\Omega_{an}\Omega_{an})=0,\qquad \tr(\chi\Omega_{an}\chi\Omega_{an})=48 A_0^4,\qquad\tr[\Omega_{an}(\chi S_{:a}- S_{:a}\chi)]=4A_0^2(2m-L_{bb})L_{aa},\\
    \tr(\chi\chi_{:a}\Omega_{an}L_{cc})=8A_0^2L_{aa}^2,\qquad\tr(\chi_{:a}\chi_{:b}\Omega_{ab})=8A_0^2(L_{aa}^2-L_{ab}L_{ab}),\qquad\tr(\chi\chi_{:a}\chi_{:b}\Omega_{ab})=0,\\
    \tr(\chi\chi_{:a}\Omega_{an;n})=12\partial_n(A_0^2)L_{aa},\qquad\tr(\chi\chi_{:a}\Omega_{ab:b})=8A_0^2(L_{aa}^2-L_{ab}L_{ab}),\qquad\tr(\chi\chi_{:a}\Omega_{bn}L_{ab})=8A_0^2L_{ab}L_{ab},\\
    \tr(\chi_{:a}E_{:a})=-16mA_0^2L_{aa},\qquad\tr(\chi_{:a}\chi_{:a}E)=8A_0^2L_{ab}L_{ab},\qquad\tr(\chi\chi_{:a}\chi_{:a}E)=0,\\
    \tr(\chi_{:aa}E)=16mA_0^2L_{aa},\qquad \tr(\chi_{:aa}\chi_{:bb})=16A_0^2L_{aa}^2,\qquad\tr(\chi_{:ab}\chi_{:ab})=16A_0^2L_{ab}L_{ab},\\
    \tr(\chi_{:a}\chi_{:a}\chi_{:bb})=0,\qquad \tr(\chi_{:b}\chi_{:aab})=-16A_0^2L_{aa}^2.
\end{gathered}
\end{equation}
In the expressions \eqref{traces1}-\eqref{traces3}, we retain only the terms that give nonzero contributions to the derivative with respect to the constant chemical potential $\mu$ included into $A_0$.

In deriving formulas \eqref{traces}-\eqref{traces3}, we have used that
\begin{equation}
\begin{gathered}
    \tr(\al^i\ga^{i_1}\cdots\ga^{i_n})=0,\qquad [\al_a,\hat{n}]=0,\qquad\partial^{\parallel}_in_i=-L_{aa},\qquad\partial_an_i\partial_bn_i=L_{ac}L_{bc},\\
    \hat{n}_{;n}=0,\qquad e_{a;n}^i=0,
\end{gathered}
\end{equation}
where $\partial^{\parallel}_i:=\partial_i-n_i\partial_n$.

\paragraph{Acknowledgments.}

We appreciate D. Fursaev and D. Vassilevich for the criticism and useful comments. The work is supported by the Ministry of Education and Science of the Russian Federation, project No 3.9594.2017/BCh.


\begin{thebibliography}{999}

\bibitem{Ditt}
W. Dittrich,
Effective Lagrangians at finite temperature,
Phys. Rev. D \textbf{19}, 2385 (1979).

\bibitem{ChodEverOw}
A. Chodos, K. Everding, D.~A. Owen,
QED with a chemical potential: The case of a constant magnetic field,
Phys. Rev. D \textbf{42}, 2881 (1990).

\bibitem{HabWeld}
H.~E. Haber, H.~A. Weldon,
On the relativistic Bose-Einstein integrals,
J. Math. Phys. \textbf{23}, 1852 (1981).

\bibitem{Weldon86}
H.~A. Weldon,
Proof of zeta-function regularization of high-temperature expansions,
Nucl. Phys. B \textbf{270}, 79 (1986).

\bibitem{ByVaZe}
A.~A. Bytsenko, L. Vanzo, S. Zerbini,
Zeta-function regularization for Kaluza-Klein finite temperature theories with chemical potentials,
Phys. Lett. B \textbf{291}, 26 (1992).

\bibitem{BytsZerb1}
E. Elizalde, S.~D. Odintsov, A. Romeo, A.~A. Bytsenko, S. Zerbini,
\textsl{Zeta Regularization Techniques with Applications}
(World Scientific, Singapore, 1994).

\bibitem{BytsZerb2}
A.~A. Bytsenko, G. Cognola, E. Elizalde, V. Moretti, S. Zerbini,
\textsl{Analytic Aspects of Quantum Fields}
(World Scientific, Singapore, 2003).

\bibitem{LoewRoj}
M. Loewe, J.~C. Rojas,
Thermal effects and the effective action of quantum electrodynamics,
Phys. Rev. D \textbf{46}, 2689 (1992).

\bibitem{ElmSkag95}
P. Elmfors, B.-S. Skagerstam,
Electromagnetic fields in a thermal background,
Phys. Lett. B \textbf{348}, 141 (1995).

\bibitem{ElmSkag98}
P. Elmfors, B.-S. Skagerstam,
Thermally induced photon splitting,
Phys. Lett. B \textbf{427}, 197 (1998).

\bibitem{Shovkovy}
I.~A. Shovkovy,
One-loop finite temperature effective potential in QED in the worldline approach,
Phys. Lett. B \textbf{441}, 313 (1998).

\bibitem{Shovk99}
I.~A. Shovkovy,
Derivative expansion of the one-loop effective action in QED,
arXiv:hep-th/9902019.

\bibitem{FursVass}
D. Fursaev, D. Vassilevich,
\textsl{Operators, Geometry and Quanta: Methods of Spectral Geometry in Quantum Field Theory}
(Springer, Heidelberg, 2011).

\bibitem{CJJThW}
A. Chodos, R.~L. Jaffe, K. Johnson, C.~B. Thorn, V.~F. Weisskopf,
New extended model of hadrons,
Phys. Rev. D \textbf{9}, 3471 (1974).

\bibitem{CJJTh}
A. Chodos, R.~L. Jaffe, K. Johnson, C.~B. Thorn,
Baryon structure in the bag theory,
Phys. Rev. D \textbf{10}, 2599 (1974).

\bibitem{DaicFr96}
J. Daicic, N.~E. Frankel,
Superconductivity of the Bose gas,
Phys. Rev. D \textbf{53}, 5745 (1996).

\bibitem{KalKaz1}
I.~S. Kalinichenko, P.~O. Kazinski,
High-temperature expansion of the one-loop free energy of a scalar field on a curved background,
Phys. Rev. D \textbf{87}, 084036 (2013).

\bibitem{KalKaz2}
I.~S. Kalinichenko, P.~O. Kazinski,
Non-perturbative corrections to the one-loop free energy induced by a massive scalar field on a stationary slowly varying in space gravitational background,
JHEP \textbf{1408}, 111 (2014).

\bibitem{KalKaz3}
I.~S. Kalinichenko, P.~O. Kazinski,
One-loop thermodynamic potential of charged massive particles in a constant homogeneous magnetic field at high temperatures,
Phys. Rev. D \textbf{94}, 125012 (2016).

\bibitem{DunHall}
G.~V. Dunne, T.~M. Hall,
Borel summation of the derivative expansion and effective actions,
Phys. Rev. D \textbf{60}, 065002 (1999).

\bibitem{DeWGAQFT}
B.~S. DeWitt,
\textsl{The Global Approach to Quantum Field Theory}, Vol. 1, 2
(Clarendon Press, Oxford, 2003).

\bibitem{BagGit1}
V.~G. Bagrov, D.~M. Gitman,
\textsl{The Dirac Equation and its Solutions}
(De Gruyter, Boston, 2014).

\bibitem{EliBorKir}
E. Elizalde, M. Bordag, K. Kirsten,
Casimir energy in the MIT bag model,
J. Phys. A \textbf{31}, 1743 (1998).

\bibitem{Furs1}
D.~V. Fursaev,
Kaluza-Klein method in theory of rotating quantum fields,
Nucl. Phys. B \textbf{596}, 365 (2001).

\bibitem{Furs2}
D.~V. Fursaev,
Statistical mechanics, gravity, and Euclidean theory,
Nucl. Phys. B (Proc. Suppl.) \textbf{104}, 33 (2002).

\bibitem{SkalDem}
V. Skalozub, V. Demchik,
Electroweak phase transition in strong magnetic fields in the standard model of elementary particles,
arXiv:hep-th/9912071.

\bibitem{KazShip}
P.~O. Kazinski, M.~A. Shipulya,
One-loop omega-potential of quantum fields with ellipsoid constant-energy surface dispersion law,
Annals Phys. \textbf{326}, 2658 (2011).

\bibitem{gmse}
P.~O. Kazinski,
Gravitational mass-shift effect in the standard model,
Phys. Rev. D \textbf{85}, 044008 (2012).

\bibitem{GrMuRaf}
W. Greiner, B. M\"{u}ller, J. Rafelski,
\textsl{Quantum Electrodynamics of Strong Fields}
(Springer, Heidelberg, 1985).

\bibitem{AkhBerQED}
A.~I. Akhiezer, V.~B. Berestetskii,
\textsl{Quantum Electrodynamics}
(Interscience Publishers, New York, 1965).

\bibitem{LandLifQED}
V.~B. Berestetskii, E.~M. Lifshitz, L.~P. Pitaevskii,
\textsl{Quantum Electrodynamics}
(Butterworth-Heinemann, Oxford, 1982).

\bibitem{DeFranc}
M. De Francia,
Free energy for massless confined fields,
Phys. Rev. D \textbf{50}, 2908 (1994).

\bibitem{KhalilB}
V.~R. Khalilov,
\textsl{Electrons in Strong Electromagnetic Fields: an Advanced Classical and Quantum Treatment}
(Gordon and Breach Sci. Pub., Amsterdam, 1996).

\bibitem{GSh}
I.~M. Gel'fand, G.~E. Shilov,
\textsl{Generalized Functions}, Vol. I: \textsl{Properties and Operations}
(Academic Press, New York, 1964).

\bibitem{Agran}
M.~S. Agranovich,
Elliptic operators on closed manifolds,
in: Contemporary Problems of Math., Fundamental Directions, VINITI, Vol. 63 (1990), pp. 5-129. [In Russian];
English transl.: Encycl. Math. Sci.,  vol. 63, Springer-Verlag, Berlin, 1994, pp. 1-130.

\bibitem{Shubin}
M.~A. Shubin,
\textsl{Pseudodifferential Operators and Spectral Theory}
(Springer, Berlin, 2001).

\bibitem{Gilkey2.11}
P.~B. Gilkey,
\textsl{Asymptotic Formulae in Spectral Geometry}
(CRC Press LLC, Boca Raton, 2004).

\bibitem{MigdalMM}
A.~B. Migdal, O.~A. Markin, I.~I. Mishustin,
The pion spectrum in nuclear matter and pion condensation,
Sov. Phys. JETP \textbf{39}, 212 (1974).

\bibitem{Furshep}
D.~V. Fursaev,
Spectral geometry of operator polynomials and applications to QFT,
arXiv:hep-th/0311080.

\bibitem{Drozd}
I. Drozdov,
Vacuum energy of quantum fields in classical background configurations,
arXiv:hep-th/0311199.

\bibitem{Dunne2}
G. Dunne,
Heisenberg-Euler effective Lagrangians: Basics and extensions,
arXiv:hep-th/0406216.

\bibitem{AndNayTran}
J.~O. Andersen, W.~R. Naylor, A. Tranberg,
Phase diagram of QCD in a magnetic field,
Rev. Mod. Phys. \textbf{88}, 025001 (2016).

\bibitem{Kapustp}
J.~I. Kapusta,
Bose-Einstein condensation, spontaneous symmetry breaking, and gauge theories,
Phys. Rev. D \textbf{24}, 426 (1981).

\bibitem{Vassil}
D.~V. Vassilevich,
Heat kernel expansion: User's manual,
Phys. Rep. \textbf{388}, 279 (2003).

\bibitem{BranGilk}
T.~P. Branson, P.~B. Gilkey,
Residues of the eta function for an operator of Dirac type,
J. Funct. Anal. \textbf{108}, 47 (1992).

\bibitem{GrRy.6}
I.~S. Gradshteyn, I.~M. Ryzhik,
\textsl{Table of Integrals, Series, and Products}
(Academic Press, Boston, 1994).

\bibitem{BerezPI}
F.~A. Berezin,
Feynman path integrals in a phase space,
Sov. Phys. Usp. \textbf{23}, 763 (1980).

\bibitem{KapustB}
J.~I. Kapusta, C.~Gale,
\textsl{Finite-Temperature Field Theory}
(Cambridge University Press, Cambridge, 2006).

\bibitem{BirDav.11}
N.~D. Birrel, P.~C.~W. Davies,
\textsl{Quantum Fields in Curved Space}
(Cambridge University Press, Cambridge, 1982).

\bibitem{BGKV}
T. Branson, P.~B. Gilkey, K. Kirsten, D.~V. Vassilevich,
Heat kernel asymptotics with mixed boundary conditions,
Nucl. Phys. B \textbf{563}, 603 (1999).

\bibitem{Ven}
A.~E.~M. van de Ven,
Index-free heat kernel coefficients,
Class. Quantum Grav. \textbf{15}, 2311 (1998).

\bibitem{BogShir}
N.~N. Bogolyubov, D.~V. Shirkov,
\textsl{Introduction to the Theory of Quantized Fields}
(Wiley, New York, 1980).

\bibitem{Miransky}
V.~A. Miransky,
\textsl{Dynamical Symmetry Breaking in Quantum Field Theories}
(World Scientific, Singapore, 1993).

\bibitem{Schwing.10}
J. Schwinger,
On gauge invariance and vacuum polarization,
Phys. Rev. \textbf{82}, 664 (1951).

\bibitem{KazMil1}
P.~O. Kazinski, V.~D. Miller,
Large mass expansion of the one-loop effective action induced by a scalar field on the two-dimensional Minkowski background with non-trivial $(1+1)$ splitting,
arXiv:1601.02486.

\bibitem{KazMil2}
P.~O. Kazinski, V.~D. Miller,
Uniquely defined one-loop effective action,
Russ. Phys. J. \textbf{59}, 1825 (2017).

\bibitem{Dyson49}
F.~J. Dyson,
The $S$-matrix in quantum electrodynamics,
Phys. Rev. \textbf{75}, 1736 (1949).

\bibitem{Zimmer}
W. Zimmermann,
The power counting theorem for Minkowski metric,
Commun. Math. Phys. \textbf{11}, 1 (1968).

\bibitem{GFSh.3}
E.~S. Fradkin, D.~M. Gitman, S.~M. Shvartsman,
\textsl{Quantum Electrodynamics with Unstable Vacuum}
(Springer, Berlin, 1991).

\bibitem{KabShab}
A. Kabo, A.~E. Shabad,
Gauge fields in a medium and a vacuum in the presence of an external potential,
Trudy FIAN \textbf{192}, 153 (1988) [in Russian].

\bibitem{DesGrigSem}
S. Deser, L. Griguolo, and D. Seminara,
Effective QED actions: Representations, gauge invariance, anomalies, and mass expansions,
Phys. Rev. D \textbf{57}, 7444 (1998).

\bibitem{McKeShub}
D.~G.~C. McKeon, C. Shubert,
A new approach to axial vector model calculations,
Phys. Lett. B \textbf{440}, 101 (1998).

\bibitem{KurVass}
M. Kurkov, D. Vassilevich,
Parity anomaly in four dimensions,
arXiv:1704.06736.


\bibitem{KontsVish}
M. Kontsevich, S. Vishik,
Determinants of elliptic pseudo-differential operators,
arXiv:hep-th/9404046.

\bibitem{GusShov}
V.~P. Gusynin, I.~A. Shovkovy,
Derivative expansion of the effective action for QED in $2+1$ and $3+1$ dimensions,
J. Math. Phys. \textbf{40}, 5406 (1999).

\bibitem{Wong}
R. Wong,
\textsl{Asymptotic Approximations of Integrals}
(SIAM, Philadelphia, 2001).

\bibitem{ReedSimon2}
M. Reed, B. Simon,
\textsl{Methods of Modern Mathematical Physics}, Vol. II: \textsl{Fourier Analysis, Self-Adjointness}
(Academic Press, New York, 1975).

\end{thebibliography}
\end{document}